\documentclass[a4paper,UKenglish,cleveref, autoref, thm-restate]{lipics-v2021}

\usepackage[T1]{fontenc}
\usepackage[utf8]{inputenc}
\usepackage{amsmath,amssymb,mathtools}
\usepackage{graphicx}
\usepackage{booktabs}
\usepackage{array}
\usepackage{enumitem}
\usepackage{tikz}
\usetikzlibrary{arrows.meta,positioning}
\usepackage{microtype}
\usepackage{url}
\usepackage{cleveref}
\usepackage{todonotes}
\usepackage[most]{tcolorbox}
\usepackage{algpseudocode}
\usepackage{amsmath}

\usepackage{marginnote}
\usepackage[strict]{changepage}

\newif\ifshowdefnotes
\showdefnotestrue

\newcommand{\rightmarginnote}[1]{%
  \checkoddpage
  \ifoddpage
    \normalmarginpar
  \else
    \reversemarginpar
  \fi
  \marginnote{\footnotesize\raggedright\textsf{#1}}%
}

\newcommand{\defmargin}[1]{%
  \ifshowdefnotes
    \rightmarginnote{#1}%
  \fi
}

\title{Integer reachability in VASS with transfers: a refined complexity analysis} 

\author{Tymoteusz Kucharek}{The University of Warsaw, Poland }{tk448520@students.mimuw.edu.pl }{https://orcid.org/0000-0002-1825-0097}{}
\author{Piotr Hofman}{The University of Warsaw, Poland }{piotr.hofman@uw.edu.pl}{https://orcid.org/0000-0001-9866-3723}{}

\authorrunning{Tymoteusz Kucharek and Piotr Hofman} 

\Copyright{Tymoteusz Kucharek and Piotr Hofman} 



\category{} 

\relatedversion{} 

\acknowledgements{This work was supported by Polish NCN grant UMO-2024/53/B/ST6/02776.}

\nolinenumbers 

\hideLIPIcs

\input{macros.sty}

\begin{document}

\maketitle

\begin{abstract}
Integer reachability is NP-complete for vector addition systems with states
(VASS), but becomes PSPACE-complete in the presence of transfer operations.
We refine this complexity gap for single-transfer VASS by identifying structural
features of transfers responsible for the increase in complexity. Each system
induces a transfer graph whose vertices are counters and whose edges represent
possible transfers. We classify its vertices as good or bad, according to the
branching and cyclic structure of their reachable subgraphs.

Let $b$ be the number of bad vertices. We show that every positive instance
admits a polynomially verifiable certificate of size $\size{I}^{O(b+1)}$, where
$\size{I}$ is the input size. Consequently, integer reachability for single-transfer
VASS can be decided in nondeterministic time $\size{I}^{O(b+1)}$; in particular, it
belongs to NP for every class with a bounded number of bad counters.

Conversely, we show that bad counters provide sufficient structural power to
encode space-bounded computation. For every transfer graph with $b$ bad
vertices, we construct a single-transfer VASS that encodes the acceptance of a
Turing machine using $b^{O(1)}$ tape cells. This yields PSPACE-hardness for
every polynomial-time constructible family of transfer graphs containing
linearly many bad vertices. Our results isolate the transfer patterns
responsible for the complexity of integer reachability.
\end{abstract}

\section{Introduction}
Vector addition systems (VASS) are one of the most important classes of infinite-state systems. They are a mathematically elegant formalism that is equivalent to Petri nets, which are used to model computer systems~\cite{DBLP:journals/topnoc/Wolf19}, business process management~\cite{DBLP:journals/sosym/Aalst15}, biology~\cite{DBLP:journals/siglog/Blondin20}, and many other domains. VASS and Petri nets have been studied extensively in theoretical computer science~\cite{DBLP:conf/ac/Esparza96,DBLP:journals/jacm/KunnemannMSSW25,DBLP:conf/focs/CzerwinskiO21}.

In this paper we consider transfer VASS, which belong to the family of VASS extensions~\cite{Bonnet2012WSTS,DBLP:conf/icalp/Valk78} introduced to model operations with a global nature, in contrast to the local modifications of pure Petri nets or ordinary VASS. The best-known example of a global operation is a reset~\cite{DBLP:journals/topnoc/AalstHHSVVW09,DBLP:conf/lics/BlondinFHMO24}.

Perhaps one of the most studied problems in computer science is reachability: given a system and two configurations, decide whether there is a run from the first configuration to the second. For ordinary VASS, the reachability problem is Ackermann-complete~\cite{DBLP:conf/focs/CzerwinskiO21,DBLP:conf/focs/Leroux21,DBLP:conf/lics/LerouxS15}, while reachability is undecidable for both reset and transfer VASS~\cite{DBLP:conf/icalp/DufourdFS98,DBLP:journals/siglog/Blondin20}. One approach to circumvent this undecidability is to consider relaxed semantics, such as integer semantics~\cite{DBLP:conf/ac/SilvaTC96,DBLP:conf/rp/HaaseH14,DBLP:journals/lmcs/BlondinHMR21}, where configurations are taken from $\Places\times\Int^d$ instead of $\Places\times\Nat^d$.

It is known that integer reachability is NP-complete for VASS and for VASS with resets~\cite{DBLP:conf/rp/HaaseH14}, while it is PSPACE-complete for VASS with transfers~\cite{DBLP:journals/lmcs/BlondinHMR21}. NP-completeness can still be considered manageable because of advances in SAT solvers, but PSPACE-completeness rules out the same practical applicability for transfer VASS that is possible for pure VASS. In this work we seek to better understand what makes transfer VASS hard. We identify a parameter that is precisely correlated with the exponent in the function bounding the length of a reachability witness.

Consider the directed graph $G$ whose vertices are counters and whose edges represent transfers in the VASS. More precisely, there is an edge from vertex $c_1$ to vertex $c_2$ if the VASS has a transition that transfers the content of counter $c_1$ to counter $c_2$. We call this graph the \emph{transfer graph}, and we use it to study the complexity of integer reachability. Vertices of this graph can be partitioned into good and bad vertices; the precise rules of partitioning are technically involved and appear in~\Cref{def:badVertices,def:counterGraphs}. The main claim is that the size of a polynomial-time verifiable witness for integer reachability depends exponentially on the number of bad vertices and polynomially on the number of good vertices. Intuitively, a vertex is good if the graph of nodes reachable from it is a path, so the choices are limited. Let $\bad(G)$ be the set of bad vertices in the graph $G$. Fix $\bad\in \Nat$, and let $\Class$ be the class of all transfer VASS whose transfer graph has at most $\bad$ bad vertices.

Our main results are captured by \Cref{thm:main:bound}, \Cref{col:main_complexity_upperbound} and \Cref{lem:main_hardness}.
\begin{itemize}
    \item 
    \Cref{thm:main:bound} establishes an upper bound on the size of a reachability witness, while~\Cref{col:main_complexity_upperbound} shows that integer reachability is in NP for transfer VASSes in $\Class$.    
    \item \Cref{lem:main_hardness} gives a matching lower bound. In particular, for any fixed transfer graph $G$ and any Turing machine using $\frac{\size{\bad(G)}}{30}$ tape cells, we construct in polynomial time a transfer VASS whose transfer graph is $G$ and two configurations such that the second configuration is reachable from the first if and only if the machine accepts.
\end{itemize} 

These results both identify transfer VASS instances for which integer reachability is tractable and clarify which transfer patterns make transfer VASS hard.

\section{Preliminaries}

We denote by $\Nat$, $\Int$\defmargin{$\Nat, \Int, [d]$} the sets of natural and integer numbers, respectively. 
Let $[d]$ be the set of all numbers from $1$ to $d$. For a set $\Set{S}$, we denote by $\size{\Set{S}}$\defmargin{$\size{X}$} the cardinality of $\Set{S}$.

If $f\colon A\to B$ is a function and $S\subseteq A$, then the \emph{restriction of $f$ to $S$} is the function
$f|_S\colon S\to B$\defmargin{$f|_S$} defined by $f|_S(x)=f(x)$ for every $x\in S$.

We denote by $\operatorname{NTIME}(\operatorname{poly}(x))$ the complexity class of problems that can be solved in nondeterministic polynomial time from $x$.

\subsection{Graphs}
A directed graph is a pair $(\Places, \Edges)$ where
$\Places$ is the set of vertices and $\Edges\subset \Places\times\Places$ is the set of edges.

For every vertex $v$ we define its \emph{open out-neighbourhood}\defmargin{$\neighbourhoodOpen{v}$} denoted as $\neighbourhoodOpen{v}$ as the set of 
vertices that are in the $\Edges$ relation with $v$, i.e. $\set{w\in \Places | (v,w)\in \Edges}$;
in particular the vertex $v$ does not belong to $\neighbourhoodOpen{v}$ unless there is an edge $(v,v)\in \Edges$.
The \emph{closed out-neighbourhood of $v$}\defmargin{$\neighbourhoodClosed{v}$} is $\neighbourhoodClosed{v} = \neighbourhoodOpen{v}\cup \set{v}$.

For a subset of the set of vertices $X\subseteq \Places$, the \emph{open} and \emph{closed out-neighbourhoods}\defmargin{$\neighbourhoodOpen{X}, \neighbourhoodClosed{X}$}
of $X$ are defined as $\neighbourhoodOpen{X} = \cup_{v \in X} \neighbourhoodOpen{v}$ and
$\neighbourhoodClosed{X} = \cup_{v \in X} \neighbourhoodClosed{v}$. 

A \emph{path} in a directed graph $G=(\Places,\Edges)$ is a finite sequence of vertices
$u_0,u_1,\dots,u_k$ such that $(u_{i-1},u_i)\in\Edges$ for every $1\le i\le k$.
The number $k$ is the length of the path. A path is \emph{simple} if all its vertices are distinct.

For vertices $u,v\in\Places$, we say that $v$ is \emph{reachable} from $u$ if there exists a path
from $u$ to $v$. In this case we also write $u\leadsto v$.
We denote by $\neighbourhoodStar{v}$\defmargin{$\neighbourhoodStar{v}$} the set of nodes reachable from $v$.

A subset $C\subseteq \Places$ of vertices is a \emph{strongly connected component} if for every two vertices $u,v\in C$, there is a path from $u$ to $v$ and a path from $v$ to $u$. In other words, all vertices in $C$ can reach each other by following directed edges.

Two vertices $u,v\in\Places$ are \emph{connected} if there exists a sequence of edges that can be traversed in either direction and that links $u$ to $v$. Equivalently, there is a path from $u$ to $v$ or from $v$ to $u$ in the underlying undirected graph obtained from $G$ by forgetting the orientation of edges.

A directed graph $G=(\Places,\Edges)$ is a \emph{cycle} if its vertices can be ordered as $v_1,\dots,v_k$ with $k\ge 1$ such that $\Edges=\{(v_i,v_{i+1})\mid 1\le i<k\}\cup\{(v_k,v_1)\}$. A \emph{simple cycle} is a cycle in which all vertices $v_1,\dots,v_k$ are distinct.

A directed graph $G=(\Places,\Edges)$ is a \emph{lasso} if it consists of a simple path $u_1,u_2,\dots,u_m$ together with a simple cycle $v_1,v_2,\dots,v_k$ such that $u_m=v_1$ and the path and cycle share no other vertices. In other words, a lasso is a simple-directed path whose last vertex is the first vertex of a simple directed cycle.

A graph $G'=(\Places',\Edges')$ is a \emph{subgraph} of $G=(\Places,\Edges)$ if $\Places'\subseteq\Places$ and $\Edges'\subseteq\Edges$. An \emph{induced subgraph} of $G$ on $\Places'$ is the graph $G[\Places']=(\Places',\Edges\cap (\Places'\times\Places'))$, containing every edge of $G$ whose endpoints both lie in $\Places'$.

For a graph $G=(\Places, \Edges)$ and a given $u \in \Places$ an \emph{out-spanning tree of $G$ rooted in $u$}
is any subgraph $H_u=(\Places', \Edges')$ of $G$, such that:
\begin{itemize}[nosep,leftmargin=1.3em]
  \item $\Places' = \Places$,
  \item for every $v \in \Places'$ there is exactly one path from $u$ to $v$. 
\end{itemize}
Notice that such a graph may not always exist, but it always exists for $\neighbourhoodStar{u}$.

\subsection{Single-transfer VASS}

\begin{definition}
For a given $d$, a transfer function is any function $[d]\to[d]$.
We denote by $\Tset_d=[d]^{[d]}$ the set of all transfer functions for a given dimension $d$.
\end{definition}

Note that transfer functions do not form a group as not every element can be reversed. For us the most important properties are that the set $\Tset$ is closed under composition, and finite. In $\Tset$ we identify a subset called:

\begin{definition}

For a given $d$, and $1\leq i,j\leq d$, let $f_i^j:[d]\to[d]$ be defined as
\[
f_i^j(k)=\begin{cases}
k & \text{for } k\neq i,\\
j & \text{for } k=i.
\end{cases}
\]
We will call such functions \emph{single-transfer functions}. 
\end{definition}

Note that such a function is also a transfer function.
Let $\Tset_d^s=\{f_i^j\mid 1\leq i,j\leq d\}$. 
This is the set of \emph{single-transfer functions} of dimension $d$.
In case $d$ is clear from the context we omit it.
Observe that by setting $i=j$ we get $f_i^j=\id$, thus $\id\in\Tset_d^s$.

\begin{definition}

  A \emph{transfer VASS} is a tuple $\V=(d,Q,\Delta)$, where:
  \begin{itemize}[nosep,leftmargin=1.3em]
    \item $d\geq 1$ is the \emph{number of counters} of $\V$,
    \item $Q$ is a finite set of elements called \emph{control space},
    \item $\Delta\subseteq Q\times\Tset_d\times\Z^d\times Q$ is a finite set of elements called \emph{transitions}.
  \end{itemize}
  The control space together with the transitions projected onto $Q\times Q$ we call the \emph{control automaton.}

\end{definition}
For a single transition $t=(p,f,v,q)\in\Delta$, let $\inn(t)\coloneqq p$, $\outt(t)\coloneqq q$, $f(t)\coloneqq f$, $\ddelta(t)=v$\defmargin{$\inn(t), \outt(t), f(t), \ddelta(t)$}.

A \emph{pure VASS} is a transfer VASS in which all transfer functions are identity i.e. for all $t\in \Delta$ it holds $f(t)=\id$. It is well known that 
integer reachability in pure VASS is an NP-complete problem. Basically, it can be encoded in the existential Presburger arithmetic~\cite{DBLP:conf/rp/HaaseH14}. 

\emph{Semantics.} A configuration is a pair $(q,v)\in Q\times\Z^d$, denoted by $q(v)$. The individual dimensions in $Z^d$ we call \emph{counters}. Note that this is not a standard definition, as usually configurations of VASS are defined as pairs of a state and a vector of natural numbers. However, in this paper, we will work with integer VASS, so we allow counters to take negative values. As the whole paper is about integer VASS, we will simply call them VASS and configurations, without the ``integer'' prefix.

For all $t\in\Delta$ we define a relation $\xrightarrow{t}_{\V}$\defmargin{$\xrightarrow{t}_{\V}$} on $Q\times\Z^d$, such that $q_1(v_1)\xrightarrow{t}_{\V}q_2(v_2)$ if and only if $\inn(t)=q_1$, $\outt(t)=q_2$ and for every $1\leq j\leq d$,
\[
v_2[j]=\left(\sum_{i\in f(t)^{-1}(\{j\})}v_1[i]\right)+\ddelta(t)[j].
\]
This means that when applying a transition $t$ to the counters, first the contents of each counter are moved to the counter specified by $f(t)$ and then $\ddelta(t)$ is added.
When $f(t)^{-1}(\{j\})$ is empty, the sum is $0$, so $t$ resets $j$-th counter to $delta(t) [j]$.

To simplify the notation whenever we consider a function $f:[d]\to [d]$ it can be lifted naturally to a function $\Z^d\to \Z^d$. Simply for a vector $v\in \Z^d$, $f(v)=v'$ such that 
$v'[j]=\sum_{i\in f^{-1}({j})}v[i]$.

Observe that $\xrightarrow{t}_{\V}$ naturally extends, for a given $w\in\Delta^*$, to a relation $\xrightarrow{w}_{\V}$. If $w=t_1\ldots t_n$, then we define $\xrightarrow{w}_{\V}$ as $\xrightarrow{t_n}_{\V}\circ\cdots\circ\xrightarrow{t_1}_{\V}$.

Moreover, we consider relations:
\begin{align*}
q_1(v_1)\xrightarrow{}_{\V}q_2(v_2) &\iff \exists t\in\Delta\ q_1(v_1)\xrightarrow{t}_{\V}q_2(v_2),\\
q_1(v_1)\xrightarrow{*}_{\V}q_2(v_2) &\iff \exists w\in\Delta^*\ q_1(v_1)\xrightarrow{w}_{\V}q_2(v_2).
\end{align*}
If the VASS $\V$ is known from the context we drop the subscript and write $\xrightarrow{t}$, $\xrightarrow{w}$, $\to$, $\to^*$\defmargin{$\xrightarrow{t}$}.

\emph{Single-transfer VASS.}
A single-transfer VASS is a transfer VASS, with the restriction that for every transition,
its associated transfer function is a single-transfer function.
In single-transfer VASS, for a transition $t=(p,f_i^j,v,q)\in Q\times\Tset_d^s\times\Z^d\times Q$, let $\src(t) = i$ and $\tgt(t) = j$.\defmargin{$\src(t), \tgt(t)$}
For the identity function $\src$ and $\tgt$ are undefined.

Reachability in Single-transfer VASSes is as hard as reachability transfer VASS, this is captured by ~\Cref{lem:simpleTransferEq}.
We use single-transfer VASS instead of general ones, because for single-transfer VASS we were able to concisely characterize parameters important for the complexity.

The \emph{size} of a transfer VASS $\V=(d,Q,\Delta)$ 
is 
$
 \size{\V}=d + |Q| + |\Delta| + \sum_{t\in\Delta} |\delta(t)| + \sum_{t\in\Delta} |f(t)|,
$
where $|\delta(t)|$ denotes the bit-length of the vector $\delta(t)$ and $\size{f(t)}$ denotes the bit-length of the representation of the transfer function $f(t)$.
For single-transfer VASS, each single-transfer function can be represented by the pair $(\src(t),\tgt(t))$ (plus special representation for $\id$), so the cost of representing $f(t)$ is $O(\log d)$.

\section{The reachability problem}\label{Sec:BigReachability}
The reachability problem for transfer VASS $\V=(d,Q,\Delta)$ is the following decision problem: 

\begin{tcolorbox}[title=\textbf{Integer Reachability Problem}]
\textbf{Input:} A single-transfer VASS $\V$ and two configurations $\initial$, $\final$.\defmargin{$\initial, \final$}
\\
\textbf{Decide:} $\initial\to^* \final$ in $\V$?\\
The input size is the size of the input VASS
plus binary encoded initial and final configurations.
\end{tcolorbox}%

As already mentioned in the introduction, this problem is PSPACE-complete, even for single-transfer VASS \cite{DBLP:journals/lmcs/BlondinR21,DBLP:journals/lmcs/BlondinHMR21}. 

The key concept in their paper we call \emph{summary functions}, defined as follows. 
\begin{definition}\label{def:summaryFunctions}
Suppose $\initial=p(\vecinitial),$ and $\final=q(\vecfinal)$\defmargin{$p(\vecinitial), q(\vecfinal)$}.
 Let  $p(\vecinitial)\xrightarrow{w}q(v)$.
We decorate each transition $w_i\in w$ with a \emph{summary function} $g_i\in \Tset$ such that
$g_i=f(w_{\size{w}})\circ f(w_{\size{w}-1})\circ \ldots \circ f(w_{i+1})$ for $i<|w|$ and
$g_{|w|}=id$. A set $\Gamma(w) = \{g_0, ..., g_{|w|}\}$ of all summary functions $g_i$ labelling a given run $w$ we call a \emph{set of summary functions induced by $w$}.
We drop the argument if the run is known from the context.
\end{definition}

It is easy to see that this labelling is equivalent to 

\begin{equation}\label{eq:composition}
g_{|w|}=id\text{ and }g_i= g_{i+1}(f(w_{i+1}))\text{ for }i<|w|  
\end{equation}

Now, observe that 
\[
v=g_0(\vecinitial)+\sum_{i=1}^{|w|} g_i(\ddelta(w_i)).
\]
 Moreover,
existence of a run 
$p(\vecinitial)\xrightarrow{*}q(\vecfinal)$ 
in the transfer VASS is equivalent to the existence of the above sum, 
where $g_i$ satisfy~\Cref{eq:composition} and $w_i$ are vectors labelling consecutive transition in the control automaton of the VASS. 

This allows us to construct a pure $d$-dimensional VASS, such that reachability in our transfer VASS
is equivalent to reachability in the pure VASS.
In the pure VASS, the states are $Q\times \Tset$, and there is a transition
 $((r,g), id, g'(v), (r',g'))$ if and only if in the transfer 
 VASS there is a transition $(r,f,v,r')$ such that 
 $g= g' \circ f$. Then in the transfer VASS there is a run $p(\vecinitial)\xrightarrow{*}q(\vecfinal)$,
if and only if there is a run from $((p,g_0),g_0(\vecinitial))$ to $((q,id),\vecfinal)$ in the pure VASS.

The constructed pure VASS is of exponential size, as $\size{\Tset}$ is exponential. The authors of~\cite{DBLP:journals/lmcs/BlondinHMR21} in their PSPACE algorithm construct the integer run in the constructed pure VASS on the fly, which is the standard trick in such algorithms\footnote{Almost, the authors of~\cite{DBLP:journals/lmcs/BlondinHMR21} do not want to start with counters $g_0(\initial)$ as this is not allowed in pure VASS, so they introduce additional dimensions and add a gadget encoding this first transfer operation. The construction presented here is sufficient for our purpose.}. 

From the above we can derive a simple lemma:

\begin{lemma}[Certificate from bounded summaries]
\label{lem:complexityFromBound}
There exists a polynomial $Poly$ such that the following holds.
Let $
    I=(V,\initial,\final)
$
be an instance of the integer reachability problem for single-transfer VASS. Suppose that there exists a run
$
    \initial\xrightarrow{t_1\cdots t_n}\final
$
whose induced set of summary functions $\Gamma(t_1\cdots t_n)$ satisfies
$
    |\Gamma(t_1\cdots t_n)| \leq m .
$
Then the reachability of $\final$ from $\initial$ in $V$ has a certificate of
size at most $Poly(m\cdot \size{I})$, verifiable in time polynomial in $m$ and $\size{I}$.
\end{lemma}
\begin{proof}
  Suppose $\initial=p(\vecinitial)$ and $\final=q(\vecfinal)$.  
  We follow the construction from~\cite{DBLP:journals/lmcs/BlondinHMR21}, more specifically Theorem 3.1 and a trick used in the proof of Theorem 5.3.
  Let $V'$, $((p,g),\vecinitial)$, and $((q,id),\vecfinal)$ be the pure VASS, and its two configurations constructed as described above. Observe that in $V'$ there is a run from $((p,g),\vecinitial)$ to $((q,id),\vecfinal)$ that visits only states
  in $Q\times \SummaryFunctions$. We define $V''$ as $V'$ restricted to 
  states in $Q\times \SummaryFunctions$, observe that its size is bounded by 
  $O(m\cdot \size{I})$. Now let $Y$ be a witness of the existence of 
  the integer run from $((p,g),\vecinitial)$ to $((q,id),\vecfinal)$
  in $V''$. Thus, in time proportional to 
  $m\cdot \size{I}$ we can guess $V''$ and $Y$ and verify them in 
  $\operatorname{Poly}(m\cdot \size{I})$ time, where $Poly$ is a fixed polynomial. \qedhere
\end{proof}

In this paper, we investigate for which transfer VASS we can show that if there is an integer run then there is one that induces the set of summary functions of size bounded by $m= Poly(k)$ where $k$ is the input transfer VASS size. 

The key contributions in this paper, proposed to characterize the well behaving subclass of transfer VASS, are definitions of transfer graph and its bad vertices.

\section{Transfer graphs and bad vertices}\label{sec:four}

Recall that for a directed graph $G=(V,E)$ and $X\subseteq V$ we use the following notation:
    \begin{itemize}
\item $G[X]$ is the subgraph of $G$ induced by $X$,
\item $\neighbourhoodClosed{X}$ is the closed out-neighbourhood of $X$, so $X\cup\{v\mid\exists u\in X\ (u,v)\in E\}$,
\item $\neighbourhoodStar{X}$ is the set of vertices reachable from $X$.
\end{itemize}
If $X$ consists of one element $u$, we use the notation $\neighbourhoodClosed{u}$ and $\neighbourhoodStar{u}$.

\begin{definition}[Bad vertices]\label{def:badVertices}
Let $G=(V,E)$ be a directed graph. We call a vertex $v\in V$ \emph{bad} if it satisfies at least one of the following conditions:
\begin{enumerate}
\item it lies on some simple cycle (no repeating vertices) of length at least $3$;
\item it has out-degree at least $2$;
\item there is a simple path from it to a bad vertex of type $1$;
\item there is a simple path from it to a different vertex which has at least $2$ out-neighbours not lying on that path.
\end{enumerate}
We denote by $\bad_i(G)$ the set of bad vertices of $G$ satisfying the $i$-th condition. Let $\bad(G)=\bad_1(G)\cup\bad_2(G)\cup\bad_3(G)\cup\bad_4(G)$.

We call a vertex $v\in V$ \emph{good} if it is not bad. We denote by $\good(G)=V\setminus\bad(G)$ the set of good vertices of $G$.
\end{definition}

\begin{definition}[Transfer graph]\label{def:counterGraphs}
  Let $\V=(d,Q,\Delta)$ be a single-transfer VASS. Observe that possible transfers associated with transitions in $\Delta$ form a graph on $[d]$, where an edge connects one counter to another if and only if there is a transition transferring from the first counter to the second one.

Formally, we denote by $G(\V)$, called the transfer graph of $\V$:
\[
G(\V)\coloneqq\left([d],\{(i,j)\mid\exists t\in\Delta\ \src(t)=i\wedge\tgt(t)=j\}\right).
\]
\end{definition}
 
Observe that $G(\V)$ depends only on the transfers of the transitions in $\Delta$.

\begin{lemma}\label{lem:simpleTransferEq}
For every instance of the reachability problem for a transfer VASS, 
there is a linear-time computable instance of the reachability 
problem for a single-transfer VASS. 
The number of counters increases by one but the number of bad vertices in the transfer graph of the single-transfer VASS can increase from $0$ to all the counters.
\end{lemma}    
The proof is moved to the appendix.

\subsection{Structure of the graph}
We begin with some basic observations regarding the structure of the induced subgraph of vertices reachable from a good vertex.

\begin{lemma}\label{lem:structure_of_the_graph}
Let $G=([d],E)$ be a graph and let $u\in\good(G)$ be a good vertex. Then $G[\neighbourhoodStar{u}]$ has the following structure:
\begin{itemize}
\item all its vertices lie on a single path, meaning we have $\neighbourhoodStar{u}=\{u=v_1,v_2,\ldots,v_n\}$ such that $(v_1,v_2)\in E,\ldots,(v_{n-1},v_n)\in E$,
\item the only other possible edges go to the previous vertex on that path (from $v_i$ to $v_{i-1}$).
\end{itemize}
An example of a graph with such properties can be seen in Fig.~\ref{fig:good-path}.
\end{lemma}

\begin{proof}
Let $H_u$ be an out-spanning tree of the graph $G[\neighbourhoodStar{u}]$ rooted in $u$.
It is trivial that it exists.
Observe that
$H_u$ must be a single branch, otherwise $u$ would be a bad vertex because it would violate the condition $4$ of definition~\ref{def:badVertices}.

Let vertices in $H_u$ be connected in the following way:
$u=v_1\to v_{2}\to v_3\to\cdots\to v_{n}$, where $n=\size{\neighbourhoodStar{u}}$.

Suppose $v_i\in \neighbourhoodStar{u}$ and $(v_i,v_j)$ is an edge in $G[\neighbourhoodStar{u}]$. We want to prove that $j=i-1$ or $j=i+1$. We prove by a contradiction. If:
\begin{itemize}
    \item $j<i-1$, then in 
$G[\neighbourhoodStar{u}]$ we have a cycle of length greater than $2$ which violates
the condition $3$ of the definition~\ref{def:badVertices}, and thus $u$ is bad;
\item $j>i+1$, then the condition $4$ of the definition~\ref{def:badVertices} is violated, and again $u$ is bad.  \qedhere
\end{itemize}
\end{proof}  

\begin{figure}[t]
\centering
\begin{tikzpicture}[>=Latex, node distance=6mm, every node/.style={draw,circle,minimum size=13mm,inner sep=0pt}]
\node (v1) {$v_1=u$};
\node[right=of v1] (v2) {$v_2$};
\node[right=of v2] (v3) {$v_3$};
\node[right=of v3] (v4) {$v_4$};
\draw[->] (v1)--(v2); \draw[->] (v2)--(v3); \draw[->] (v3)--(v4);
\draw[->] (v2) to[bend left=40] (v1);
\draw[->] (v4) to[bend left=40] (v3);
\end{tikzpicture}
\caption{Example $G[\neighbourhoodStar{u}]$ for a good vertex $u$.}
\label{fig:good-path}
\end{figure}

\section{Upper bound on integer reachability parametrized by transfer graph} \label{sec:upperbound}

The goal of this section is to prove the following theorem, which allows us to define the number $m$ in the statement of~\Cref{lem:complexityFromBound},
as a function of the numbers of good and bad vertices.

\begin{theorem}\label{thm:main:bound}
Let $\V$ be a single-transfer VASS, let $\initial, \final$ be its two 
configurations, and $\initial\xrightarrow{t_1,\ldots t_n}\final$ be an 
integer run. Let $g_0, g_1,\ldots g_n$ be a sequence of the 
summary functions 
for the run $t_1,\ldots,t_n$ and 
$\SummaryFunctions =\set{g_0,\ldots,g_n}$. Then
\[
\size{\SummaryFunctions}
\leq\left(1-|\good(G)|+
\sum_{v\in\good(G)}|\neighbourhoodStar{v}|\right)\cdot\prod_
{v\in\bad(G)}\size{\neighbourhoodStar{v}
\cap \neighbourhoodClosed{\bad(G)}}.
\]
\end{theorem}

Note that the bound above is independent from $n$.

From~\Cref{thm:main:bound} and \Cref{lem:complexityFromBound} we get the following corollary:

\begin{corollary}\label{col:main_complexity_upperbound}
    For any fixed $\bad\in \Nat$ let $\Class$ be a class of transfer graphs with at most $\bad$ bad vertices.
    The integer reachability problem is in NP for single-transfer VASSes with transfer graphs in $\Class$.
\end{corollary}
\begin{proof}
    Let $(\V, \initial, \final)$ be an instance of the problem, let $d$ be the number of counters of $V$.
    From ~\Cref{thm:main:bound} we can bound $\size{\Gamma}$ by $m = d^{2+B}$, which is polynomial in $\size{V}$.
    By ~\Cref{lem:complexityFromBound} there exists a reachability certificate of size polynomial in $m \cdot \size{V}$.
    \qedhere
\end{proof}

For ~\Cref{thm:main:bound} the proof strategy is to first consider what happens for good nodes, and then use it to prove the theorem. This is captured by~\Cref{lem:all_good_vertexes}; the~\Cref{lem:single_good_vertex} is an auxiliary lemma.

\begin{lemma}\label{lem:single_good_vertex}
Let $\V$ be a single-transfer VASS and let $G = G(\V)$. Let $t_1\ldots t_n$ be a run and $f_1,\ldots, f_n$ be the induced sequence of simple transfers ($f_i=f(t_i)$), and $g_0,\ldots g_n$ be a summary sequence of $t_1,\ldots t_n$.

Then, for any $v\in\good(G)$, the sequence $g_0(v), g_1(v),\ldots,g_{n-1}(v),g_{n}(v)=\id$ 
changes its value at most $|\neighbourhoodStar{v}|-1$ times.
\end{lemma}

\begin{lemma}\label{lem:all_good_vertexes}
Under the same assumptions as in~\Cref{lem:single_good_vertex} the sequence
    $
g_{0}|_{\good(G)},\ldots,g_{n}|_{\good(G)}
$
changes its value at most
\[
\sum_{v\in\good(G)}|\neighbourhoodStar{v}|-|\good(G)|\text{ times.}
\]
\end{lemma}
\begin{proof}[Proof of~\Cref{lem:all_good_vertexes}]
Each change of the restricted sequence 
$g_{0}|_{\good(G)},g_{1}|_{\good(G)},\ldots 
g_{n}|_{\good(G)}$ corresponds to a change, 
in at least one of the sequences 
$g_{0}(v),g_{1}(v),\ldots g_n(v)$ for some 
$v\in\good(G)$.

Thus, by contradiction, if there were more than
 $\sum_{v\in\good(G)}|\neighbourhoodStar{v}|-
 |\good(G)|$ changes in the restricted sequence, 
 by the pigeonhole principle, there would be 
 at least $|\neighbourhoodStar{v}|$ 
 changes in the sequence for some $v$. 
 But by~\Cref{lem:single_good_vertex} 
 there are at most $|\neighbourhoodStar{v}|-1$ 
 changes. \qedhere
\end{proof}    

\begin{proof}[Proof of \Cref{lem:single_good_vertex}]
Let $v_1=v,v_2,\ldots,v_k$ 
be the path described in~\Cref{lem:structure_of_the_graph}, such that $\neighbourhoodStar{v}=\{v_1,\ldots,v_k\}$ and, for any $j\in [k-1]$, the only edges are from $v_j$ to $v_{j+1}$ and possibly backwards, from $v_{j+1}$ to $v_j$.

We introduce a linear order on the set $\set{v_1,\ldots v_k}$, namely: 
\[
v_1\preceq v_2\preceq v_3\preceq \cdots\preceq v_k.
\]

First we prove that all $g_i$ satisfy the following monotonicity property:
\begin{claim}\label{clm:monotonicity_one}
If $v_j\preceq v_{j'}$ then $g_i(v_j)\preceq g_i(v_{j'})$.
\end{claim}
\begin{proof}
    We prove this claim via downward induction over $i$. For $i=n$ the summary function $g_n$ is the identity, so the claim holds.

    Suppose it holds for $i$; we prove it for $i-1$. $g_{i-1}=g_i \circ f_i$.
    Take $v_j\prec v_{j'}$ (We do not consider the trivial case when $v_j= v_{j'}$). Observe that since $f_i$ is a single transfer we have that $f_i(v_j)\preceq f_i(v_{j'})$; this is because of~\Cref{lem:structure_of_the_graph} and the following case analysis (here we use the single-transfer property):
    \begin{itemize}
        \item $f_i(v_j) \in \set{v_{j-1},v_j}$ and $f_i(v_{j'})\in \set{v_{j'},v_{j'+1}}$, then it is trivial.
        \item $f_i(v_j) = v_{j+1}$ and $f_i(v_{j'})= v_{j'}$, then as $j+1\leq j'$ the observation holds.
        \item $f_i(v_j) = v_{j}$ and $f_i(v_{j'})= v_{j'-1}$, then as $j\leq j'-1$ the observation holds.
    \end{itemize}

    Now, we use the inductive assumption for $g_i(v_j)$ and $g_i(v_{j'})$ and get $g_{i-1}(v_j)=g_{i}(f_i(v_j))\preceq g_{i}(f_i(v_{j'}))=
    g_{i-1}(v_{j'})$, as required. \qedhere
\end{proof}

Using the claim, we show that:
\begin{claim}\label{clm:monotonicity_two}
For any $i\in [n-1]$ it holds $g_{i+1}(v)\preceq g_{i}(v)$.
\end{claim}
Notice that this claim is specifically about $v$, it might not hold for other $v_j$.
\begin{proof}
We prove this via downward induction over $i+1$.
For $i+1=n$ we have that $g_n(v)=v$ which is the smallest element in the order, so the induction hypothesis holds. 

We assume that the hypothesis holds for $i+1$, and we prove it for $i$.
Note that $g_{i}=g_{i+1} \circ f_{i+1}$. 
According to~\Cref{lem:structure_of_the_graph}, $f_{i+1}(v)=v=v_1$ or $f_{i+1}(v)=v_2$. We consider these two cases:
\begin{itemize}
    \item If $f_{i+1}(v)=v=v_1$ then $g_{i}(v)=g_{i+1}(f_{i+1}(v))=g_{i+1}(v)$ so,
    as required, $g_{i}(v)\succeq g_{i+1}(v)$.
    \item If $f_{i+1}(v)=v_2$ then $g_{i}(v)=g_{i+1}(f_{i+1}(v))=g_{i+1}(v_2)$.
    But now we use~\Cref{clm:monotonicity_one} for $g_{i+1}$ and $v\prec v_2$ and get
    $g_{i}(v)=g_{i+1}(v_2)\succeq g_{i+1}(v)$ as required.  \qedhere
\end{itemize}
\end{proof}

As the sequence $g_{i}(v)$ is monotone (according to~\Cref{clm:monotonicity_two}), it can change its value at most $|\neighbourhoodStar{v}|-1$ times. \qedhere
\end{proof}

\begin{proof}[Proof of~\Cref{thm:main:bound}]
Let $0 \le i_1<\ldots<i_m < n$ be the indices at which the restriction of the sequence
$g_{0},\dots,g_{n}$ to $\good(G)$ changes, namely all indices such that
$g_{i_j}|_{\good(G)}\neq g_{i_j+1}|_{\good(G)}$.
Observe that with this definition the restriction is constant on every interval
$i_j<l\le i_{j+1}$ (assuming $i_0 = -1$ and $i_{m+1} = n$).

We claim that:
\begin{claim}
For every $v\in\neighbourhoodClosed{\bad(G)}$, every $0\le j \le m$, and every
$i_j<l\le i_{j+1}$, we have
\[
g_{l}(v)\in g_{i_{j+1}}\bigl(\neighbourhoodStar{v}\cap \neighbourhoodClosed{\bad(G)} \bigr).
\]
\end{claim}
\begin{proof}
We consider two cases: $v\in  \neighbourhoodClosed{\bad(G)} \setminus \bad(G)$ or 
$v\in \neighbourhoodClosed{\bad(G)} \cap \bad(G)$.
 
\begin{enumerate}
    \item In the first case note that $v$ is a good vertex. By the definition of the sequence of $i$-s, $g_l$ and $g_{i_{j+1}}$ are equal on good vertices, in particular on $v$, so we only need to know that $v\in \neighbourhoodStar{v}\cap \neighbourhoodClosed{\bad(G)}$ which trivially holds.
    \item The second case we show by contradiction. Suppose the claim fails for some $l$ in the interval $i_j<l\le i_{j+1}$. We choose the largest $l$ and a vertex $v$ such that $g_l(v)\notin g_{i_{j+1}}\bigl(\neighbourhoodStar{v}\cap \neighbourhoodClosed{\bad(G)} \bigr)$. Then $l<i_{j+1}$, since 
    $v\in\neighbourhoodStar{v}\cap \neighbourhoodClosed{\bad(G)}$. 
    
    Moreover, 
$g_{l}(v)\neq g_{l+1}(v)$, so $f_{l+1}(v)=w$ for some $w$ with
$g_{l}(v)=g_{l+1}(w)$. Note that, since $v$ is bad, we have $w\in\neighbourhoodClosed{\bad(G)}$.
We use the maximality of $l$, and conclude that 
\[
g_{l+1}(w)\in
g_{i_{j+1}}(\neighbourhoodStar{w}\cap \neighbourhoodClosed{\bad(G)}).
\]
As $\neighbourhoodStar{w}\subseteq \neighbourhoodStar{v}$ we get that
\[g_{l}(v)=g_{l+1}(w)\in
g_{i_{j+1}}(\neighbourhoodStar{w}\cap \neighbourhoodClosed{\bad(G)})\subseteq g_{i_{j+1}}(\neighbourhoodStar{v}\cap \neighbourhoodClosed{\bad(G)}),
\]
as required.  \qedhere
    
\end{enumerate}
\end{proof}

Therefore, the sequence $g_{i_j+1},\dots,g_{i_{j+1}}$ is constant on $\good(G)$ and each of its values on $\bad(G)$ is determined by one of the finitely many choices in
\[
\prod_{v\in\bad(G)} g_{i_{j+1}}(\neighbourhoodStar{v}\cap \neighbourhoodClosed{\bad(G)}).
\]
Hence,
\[
\left|\{g_{i_j+1},\dots,g_{i_{j+1}}\}\right|
\le 
\prod_{v\in\bad(G)} \left|g_{i_{j+1}}(\neighbourhoodStar{v}\cap \neighbourhoodClosed{\bad(G)})\right|.
\]
By \Cref{lem:all_good_vertexes}, we have
\[
m\le \sum_{v\in\good(G)}|\neighbourhoodStar{v}|-|\good(G)|.
\]
Combining the previous estimates yields
\begin{align}
|\SummaryFunctions|
&\le (m+1)\prod_{v\in\bad(G)} \left|\neighbourhoodStar{v}\cap \neighbourhoodClosed{\bad(G)}\right| \le  \\
& \left(1-|\good(G)|+\sum_{v\in\good(G)}|\neighbourhoodStar{v}|\right) 
\prod_{v\in\bad(G)} \left|\neighbourhoodStar{v}\cap \neighbourhoodClosed{\bad(G)}\right|.
\end{align}
This is exactly the claimed bound. \qedhere
\end{proof}

\section{Lower bound on integer reachability parametrized by transfer graph}\label{sec:six}

In this section, we show that for a given graph $G$ we can encode runs of a Turing machine over a binary alphabet, restricted to space of size linear in $|\bad(G)|$, by runs of a single-transfer VASS such that its transfer graph is precisely $G$.
This, together with ~\Cref{sec:upperbound}, shows that the number of bad vertices is indeed directly responsible
for the hardness of integer reachability parametrized by the transfer graph.  

\begin{definition}[single-tape Turing machine]
A \emph{single-tape Turing machine} is a tuple
$\mathcal{M} = (Q,\set{0,1},\delta,q_0,q_{\mathrm{acc}},q_{\mathrm{rej}})$,
where: $Q$ is a finite set of states; $\set{0,1}$ is a tape alphabet; $q_0 \in Q$ is the initial state; $q_{\mathrm{acc}}, q_{\mathrm{rej}} \in Q$ are distinct accepting
    and rejecting states;
    \[
    \delta \subseteq
    (Q \setminus \{q_{\mathrm{acc}},q_{\mathrm{rej}}\}) \times \set{0,1}
    \times Q \times \set{0,1} \times \{L,R\}
    \]
    is the transition function.

The machine has one tape, infinite in both directions, initially containing the input word $w \in 0^{\omega}$. In one transition,
if $(q,a,q',b,D)\in \delta$, then the machine, being in state $q$ and with the head pointing to some memory cell on the tape with the value
$a$, writes $b$ to this memory cell, changes the state to $q'$, and moves the tape head one cell
in direction $D \in \{L,R\}$ left or right.
\end{definition}

\newcommand{\BoundedMachineProblem}[2]{%
\begin{tcolorbox}[title=\textbf{Space-Bounded-TM-Acceptance}]
\textbf{Input:} A Turing machine $#1$,  
and $1^{#2}$, where $#2 \in \Nat$.\\
\textbf{Decide:} Does the machine $#1$ accept the 
empty word using at most $#2$ cells of the working tape?
\end{tcolorbox}%
}

\BoundedMachineProblem{M}{n}

The following standard PSPACE-completeness result follows
from the characterization of PSPACE by polynomial-space
Turing machines; see, e.g., \cite[Section~8.3]{Sipser2012}.

\begin{theorem}\label{thm:PSPACE_Hardness}
$\textsc{Space-Bounded-TM-Acceptance}$ is PSPACE-complete under
polynomial-time many-one reductions.
\end{theorem}

The main result of this section is as follows:

\begin{lemma}\label{lem:main_hardness}
Let $G$ be a graph with $n\in \Nat$ bad vertices, and let $(M, m)$ be the input to the Space-Bounded-TM-Acceptance problem with $m \leq \frac{n}{\Bigconst}$. Then in polynomial time we can construct a transfer VASS $\V$ with two configurations $\initial$ and $\final$
such that: the transfer graph of $\V$ is $G$ and $\initial\xrightarrow{{*}}\final$ if and only if the empty word is accepted by the machine $M$.
\end{lemma}

The proof of the above lemma is postponed after~\Cref{col:GN_reachability_hardness}, stated below. \Cref{col:GN_reachability_hardness} is a consequence of \Cref{lem:main_hardness} and \Cref{thm:PSPACE_Hardness}.

\begin{corollary}[Hardness for every fixed graph family]\label{col:GN_reachability_hardness}
Let $(G_n)_{n\geq 1}$ be a polynomial-time constructible family of
graphs such that $G_n$ has at least $n$ bad vertices for every
$n\geq 1$. Let the $G_n$-integer reachability problem be the reachability problem where instances are restricted to transfer VASSes with induced transfer graphs in the family $G_n$.

\noindent Then the $G_n$-integer reachability problem is PSPACE-hard.
\end{corollary}

\begin{proof}[Proof of~\Cref{col:GN_reachability_hardness}]
For every instance of the Space-Bounded-TM-Acceptance problem $(M,1^n)$, by~\Cref{lem:main_hardness}, we can construct in polynomial time a VASS $\V$ and its two configurations $\initial, \final$ such that:
{\it (i)} the transfer graph of $\V$ is $G_{\Bigconst n}$ and
{\it (ii)}  $\initial\xrightarrow{*}\final$ if and only if the empty word is accepted by $M$ using at most $n$ memory cells of the working tape.
As the Space-Bounded-TM-Acceptance problem is PSPACE-hard~(\Cref{thm:PSPACE_Hardness}), we conclude that the $G_n$-integer reachability problem is also PSPACE-hard. \qedhere
\end{proof}

\subsection{The Proof of~\Cref{lem:main_hardness}}
This VASS is constructed in three steps.
First, we find a large subgraph $G'$ of $G$ with some desired properties (we call such graphs clean-simple).
Next we show ~\Cref{lem:main_hardness} with a better constant for clean-simple graphs, obtaining a VASS $V'$.
Finally, we modify $V'$ so that its transfer graph is $G$ instead of $G'$ but it still satisfies ~\Cref{lem:main_hardness}.

We begin with the last step, as it is the easiest.

\begin{lemma}\label{lem:remove_transfers}
    Let $\V'=(d,Q',\Delta')$ be a single-transfer VASS with the transfer graph $G'=([d],E')$.
    Let $G=([d],E)$ be a graph such that $E'\subseteq E$. Then there is a polynomial-time constructible transfer VASS $\V=(d,Q,\Delta)$ with a transfer graph $G$ such that $Q'\subseteq Q$ and for any configurations $(p,x),(q,y)\in Q'\times \Int^d$ it holds that $(p,x)\xrightarrow{*}_{\V'}(q,y)$ if and only if 
    $(p,x)\xrightarrow{*}_{\V}(q,y)$.
\end{lemma}
\begin{proof}
    We simply add new states and transitions between them that allow for missing transfers, and do not modify counters in any other way.
    Because the newly added elements of the control automaton are not connected to the $Q'$ part, they
     do not affect the reachability between configurations with states in $Q'$. \qedhere
\end{proof}

\subsubsection{The simplification}

This section begins by defining simple graphs (~\Cref{def:simple_graph}) and clean-simple graphs (~\Cref{def:simple_and_clean}), the latter being precisely the subgraphs that we look for.
\Cref{lem:simple_graph} says that we can always find a large simple subgraph and \Cref{lem:clean_simple_graph} says that within a simple graph we can always find a big clean-simple subgraph.

The section concludes with~\Cref{lem:main_lemma_simple}, a reformulation of~\Cref{lem:main_hardness} for clean-simple transfer graphs with better constants.
Together with~\Cref{lem:simple_graph,lem:clean_simple_graph,lem:remove_transfers}, it gives~\Cref{lem:main_hardness}.

\begin{definition}\label{def:simple_graph}
A graph is \emph{simple} if each bad vertex $x$ satisfies the following properties:
\begin{itemize}
    \item if $x$ is of type $1$ then $N^*(x)$ is a simple cycle,
    \item if $x$ is of type $2$ then $N^*(x)=N[x]$, and it contains two nodes, that are good,
    \item if $x$ is of type $3$ then $N^*(x)$ is a simple lasso,
    \item if $x$ is of type $4$ but not of type $2$ then $\size{N(x)}=1$. Its out-neighbour is a bad vertex of type $2$ or $4$.
\end{itemize}
Good vertices do not have outgoing edges.
\end{definition}

\begin{definition}\label{def:simple_and_clean}
We say that a simple graph is clean if its good vertices can be partitioned into two sets, green and red, in such a way that every bad vertex of type $2$ has one red and one green neighbour.
\end{definition}    

\begin{lemma}\label{lem:simple_graph}
Let $G$ be a graph with $n$ bad vertices. It contains a simple subgraph $G'$ with $\frac{n}{3}$ bad vertices.
\end{lemma}

The proof of~\Cref{lem:simple_graph} is conceptually simple but laborious: we clean the graph $G$
in several steps until a simple subgraph is obtained. The full proof is given in
~\Cref{app:two}.

\begin{lemma}\label{lem:clean_simple_graph}
Let $G$ be a simple graph with $n$ bad vertices. It contains a clean-simple subgraph $G'$ with $\frac{n}{2}$ bad vertices.
\end{lemma}

The proof of the above lemma is also straightforward and can be found in~\Cref{app:two}.

\begin{lemma}\label{lem:main_lemma_simple}
Let $G$ be a clean-simple graph with $n\in \Nat$ bad vertices, and let $(M,1^{\frac{n}{5}})$ be the input to the Space-Bounded-TM-Acceptance problem. Then in polynomial time we can construct a transfer VASS $\V$ with two configurations $\initial$ and $\final$
such that: the transfer graph of $\V$ is $G$ and $\initial\xrightarrow{{*}}\final$ if and only if the empty word is accepted by the machine $M$ using at most $\frac{n}{5}$ work memory cells.
\end{lemma}

The proof of~\Cref{lem:main_lemma_simple} is postponed to the next subsection. Here we only present how to prove~\Cref{lem:main_hardness} using it and~\Cref{lem:remove_transfers,lem:simple_graph,lem:clean_simple_graph}.

\begin{proof}[Proof of~\Cref{lem:main_hardness}]
Let $G'$ be a clean-simple subgraph of $G$ with at least $\frac{n}{6}$ bad vertices. It exists because of~\Cref{lem:simple_graph} and~\Cref{lem:clean_simple_graph}, composed together.
Using~\Cref{lem:main_lemma_simple} we can construct in polynomial time a
transfer VASS $\V'$ with a transfer graph $G'$ and two configurations $\initial$ and $\final$ such that $\initial\xrightarrow{*}\final$ if and only if the machine $M$ accepts the empty word using at most $\frac{n}{\Bigconst}$ memory cells.
Then we produce the required $\V$ using~\Cref{lem:remove_transfers}. \qedhere
\end{proof}

\subsubsection{Proof of~\Cref{lem:main_lemma_simple}}

We start by introducing a new intermediate model called an \emph{automaton with memory}, equivalent to a
Turing machine with linearly bounded tape. It is more convenient to encode reachability problem for such an automaton within integer reachability of transfer VASS, compared to the direct encoding of linearly bounded Turing machines. For the moment, let us disregard that the transfer VASS to be constructed must have a given transfer graph.
The integer run in the transfer VASS will encode the run of the automaton with memory, however a single step of the automaton will be matched by a sequence of moves of the transfer VASS. The intuition is that there will be a bijection between a subset of counters and cells of the memory in the automaton with memory. The value stored in such a counter, $1$ or $-1$, will match the value stored in the corresponding cell of the memory. Maintaining the control state and corresponding values of counters is easy assuming that we have a way to correctly recognize the current configuration of the transfer VASS. Note that for the above we do not need to use transfer transitions, thus this part is independent from the transfer graph. Transfers are used to recognize the current configuration.
As in the transfer VASS we do not have any test, the recognition of the current configuration is resolved by nondeterminism, but after each guess the transfer VASS goes into the testing mode that uses transfers to verify if the guess was correct, if not the transfers will accumulate tokens in some special counters and these accumulated tokens will stay there forever. Thus any incorrect guess automatically will make reaching the final marking impossible. The testing gadgets are described in~\Cref{prop:meta}~and~\Cref{prop:cycle,prop:Y,prop:lasso}. Different testing gadgets are compatible with bad vertices of different types, so the particular combination of testing gadgets used in the VASS depends on the given transfer graph.

\begin{definition}[automaton with memory]
An automaton with memory is a 5-tuple $(n,Q\times{\set{?,-1,1}}\times \set{1\ldots n},T,c_0,c_f)$ where
$n\in \Nat$ is the number of memory cells,
$Q\times{\set{?,-1,1}}\times \set{1\dots n}$ is a finite set of states,
$T$ is the set of transitions of two types:
\begin{itemize}
    \item test transitions $(q,?,i)\xrightarrow{j} (q,j,i)$ where $j\in \set{-1,1}$, are used to check the content of the memory cell number $i$ and if its value is $j$ then we record this in the second component of the place. Memory is not changed when we make this transition.
    \item update transitions $(q,j,i)\xrightarrow{l} (p,?,k)$
    where $j,l\in \set{-1,1}$. When performed, it simply modifies the value of the memory cell number $i$ to $l$, and it changes the control state to $(p,?,k)$.
\end{itemize} 
$c_0,c_f\in ((n,Q\times{\set{?,-1,1}}\times \set{1\ldots n})\times \set{-1,1}^{n})$ are the initial and final configurations.
The semantics is natural and explained in the description of
transitions.
\end{definition}

Note that transitions are designed in such a way that updating moves are interleaved with test moves.

\begin{tcolorbox}[title=\textbf{Reachability Problem for automata with memory}]
\textbf{Input:} An automaton with memory $M=(n,Q\times{\set{?,-1,1}}\times \set{1\ldots n},T,c_0,c_f)$.\\
\textbf{Decide:} $c_0\to^* c_f$ in $M$?
\end{tcolorbox}%

\begin{theorem}\label{thm:hardness_of_memory_automata}
There is a polynomial-time algorithm that takes as input
an instance of the Space-Bounded-TM-Acceptance problem $M,1^n$ and produces for it an instance of Reachability Problem for automata with memory of size $n$, such that answers for both instances agree.
\end{theorem}

Note that a consequence of this is that the Reachability Problem for automata with memory is PSPACE-hard. The proof of~\Cref{thm:hardness_of_memory_automata} is a trivial exercise and was moved to the appendix.
Now, to prove~\Cref{lem:main_lemma_simple} it is sufficient to show:
\begin{claim}\label{clm:last_claim}
There is a polynomial-time algorithm which takes as input an instance $M$ of
the reachability problem for automata with memory of size $\frac{n}{5}$ and a
clean-simple graph $G$ with $n$ bad vertices, and outputs a transfer VASS $\V$
with transfer graph $G$ and two configurations $\initial,\final$ such that
$
    \initial \xrightarrow{*}_{\V} \final
$
if and only if the instance $M$ is accepting.
\end{claim} 

Since the key idea of the construction is nondeterministic guesses of the current configuration and mechanism to penalize incorrect guesses, we need to have some semi-invariant that is preserved along the correct run and any incorrect step breaks the semi-invariant so that it cannot be repaired. This semi-invariant is captured by the following definition.

\begin{definition}
    Let $C_{Pos},C_{Neg},C_{Cell}$ be three disjoint sets of counters. We say that a counter valuation $v$ is \emph{considerable} if $v(c)\geq 0$ for all $c\in C_{Pos}$,
    $v(c)\leq 0$ for all $c\in C_{Neg}$, and 
    $v(c)\in \set{-1,1}$ for all $c\in C_{Cell}$.
    We also say that a considerable counter valuation is 
    valid if $v(c)=0$ for all $c\in C_{Pos}\cup C_{Neg}$. 
\end{definition}

We extend the above definition to configurations, saying that a configuration is considerable/valid if the vector of its counters is considerable/valid.

Conceptually, we design the transfer VASS in such a way that it starts from a valid configuration and:
\begin{itemize}
    \item all reachable configurations (except some intermediate configurations inside gadgets described in~\Cref{prop:cycle,prop:lasso}) are considerable, 
    \item from any reachable configuration (except the intermediate configurations inside gadgets) that is not valid only invalid configurations can be reached,
    \item the target configuration is valid,
    \item if the final configuration of the automaton with memory can be reached then a valid configuration with a specific final state can be reached in the transfer VASS.
\end{itemize} 

The hard part of the construction is the gadgets corresponding to test transitions in the automaton with memory. We explain them now. All the technical constructions below follow the same pattern.

\begin{proposition}\label{prop:meta}[Meta proposition]
Let $G$ be a graph of some type, with a particular split of vertices into $C_{Pos},C_{Neg},C_{Cell}$, and $x\in C_{Cell}$. Then, in polynomial time, it is possible to construct a transfer VASS
with $3$ special states $p_{init},p_{-1},p_{1}$ such that:
\begin{itemize}
        \item  $G$ is its transfer graph,
        \item  {\bf (error propagation)}, if $v$ is a considerable counter valuation, but it is not valid  and from $(p_{init},v)$ it is possible to reach a configuration $(p_{-1},v')$ or $(p_{1},v')$ then $v'$ is not valid as well.
        \item  {\bf (value $-1$ guessing)}, if $v$ is a valid counter valuation such that $v(x)=-1$ and from 
        $(p_{init},v)$ it is possible to reach $(p_j,v')$, where $j\in \set{-1,1}$ then $v'$ is a valid counter valuation if and only if $v'=v$ and $j=-1$.
        \item  {\bf (value $1$ guessing)}, if $v$ is a valid counter valuation such that $v(x)=1$ and from 
        $(p_{init},v)$ it is possible to reach $(p_j,v')$, where $j\in \set{-1,1}$ then $v'$ is a valid counter valuation if and only if $v'=v$ and $j=1$.
    \end{itemize}   
\end{proposition}

Since all the constructions follow the same pattern, we specify only the graph, the splitting, and $x$.
\begin{proposition}\label{prop:cycle}
    The meta proposition~(\Cref{prop:meta}) for:
    Let $G$ be a cycle, with vertices $c_1,c_2,\ldots, c_{3k+r}$, where $0<k\in \Nat$ and $r<3$.
    \begin{itemize}
        \item     $C_{Pos}=\set{c_i|i=3j+1, j<k} \cup \set{c_i|i>3k}$,
        \item $C_{Neg}=\set{c_i|i=3j+3, j<k}$,
        \item $C_{Cell}=\set{c_i|i=3j+2, j<k}$.
    \end{itemize}
    We take $x\in C_{Cell}$.
\end{proposition}

\begin{proof}
The VASS has three parts $p_{init}$, Part $-1$ and Part $1$. In $p_{init}$ we choose whether $v(x)=-1$ or $v(x)=1$ and depending on this the VASS goes to the corresponding part. Each part is responsible for verifying whether the initial choice was correct.
Let $n=3k+r$, and $x=c_e$.

\noindent{\bf Part $-1$.}
 The control automaton is a path $p_{1}^{-1}, p_{2}^{-1}, p_{3}^{-1}\ldots p_{n\cdot (n-1)+1}^{-1}, p_{n\cdot (n-1)+3}^{-1}=p_{-1}$ that 
performs the following sequence of operations:

\begin{algorithmic}
\State $(p_{1}^{-1},\id, 
        c_e\gets c_e +1,
        p_{2}^{-1})$
\State $\ell \gets 2$
\For{$i \gets 0$ \textbf{to} $n-1$}
    \For{$j \gets 1$ \textbf{to} $n-1$}
        \State $(p_{\ell}^{-1},
        c_{(e+i-j)\bmod n}\ \text{transfer to}\
        c_{(e+i-j+1)\bmod n},
        \mathbf{0}, 
        p_{\ell+1}^{-1})$
        \State $\ell \gets \ell + 1$
    \EndFor
\EndFor
\end{algorithmic}

As we assume $v(c_e)=-1$ we add $1$ to it, and then we start to do cyclic shifts, until the counter $v(c_{e-1})$ will not be transferred back to itself. As we cyclic shift backward we are guaranteed that values of all counters except $c_e$ and $c_{e-1}$ are preserved. The value of $c_e$ is $0$ afterwards, and the new value of $c_{e-1}$ equals $v(c_{e-1})+ v(c_e)+1$.
Note that $c_{e-1}\in C_{Pos}$ and $v(c_e)\geq -1$ so
$v(c_{e-1})+ v(c_e)+1\geq 0$. Also, $v(c_{e-1})+ v(c_e)+1$ equals $0$ only if 
$v_{e-1}=0$ and $v_{e}=-1$, so the guess was correct.

What remains is to repair the value of $v(c_e)$ to $-1$, so we finish with:

\begin{algorithmic}
\State $(p_{n(n-1)+2}^{-1},\id,
        c_e\gets c_e-1,
        p_{n(n-1)+3}^{-1})$
\end{algorithmic}

\noindent{\bf Part $1$.}
Here the situation is much simpler.
The control automaton is a path $p_{1}^{1}, p_{2}^{1}, p_{3}^{1}, p_{4}^{1} = p_{1}$ that 
performs the following sequence of operations:

\begin{algorithmic}
\State $(p_{1}^{1},\id,
        c_e\gets c_e -1,
        p_{2}^{1})$
\State $(p_{2}^{1},
        c_{e}\ \text{transfer to}\ c_{e+1 \bmod n},
        \mathbf{0},
        p_{3}^{1})$
\end{algorithmic}

As we assume $v(c_e)=1$ we subtract $1$ from it, and then we transfer its contents to the next counter on the cycle.
The value of $c_e$ is $0$ afterwards, and the new value of $c_{e+1}$ equals $v(c_{e+1})+ v(c_e)-1$.
Note that $c_{e+1}\in C_{Neg}$ and $v(c_e)\leq 1$ so
$v(c_{e+1})+ v(c_e)-1\leq 0$. Also, $v(c_{e+1})+ v(c_e)-1$ equals $0$ only if 
$v(c_{e+1})=0$ and $v(c_{e})=1$, so the guess was correct.

What remains is to repair the value of $v(c_e)$ to $+1$, so we finish with:

\begin{algorithmic}
\State $(p_{3}^{1},\id,
        c_e\gets c_e+1,
        p_{4}^{1})$  \qedhere
\end{algorithmic}
\end{proof}

\begin{proposition}\label{prop:Y}
    The meta proposition~(\Cref{prop:meta}) for:
    Let $G$ be a $Y$ shape, with vertices $l_1,l_2,\ldots, l_m,c_l,c_r$ where $l_1,l_2,\ldots, l_m,c_l$ and $l_1,l_2,\ldots, l_m,c_r$ form two paths.   
    \begin{itemize}
        \item $C_{Pos}=\set{c_r}$,
        \item $C_{Neg}=\set{c_l}$,
        \item $C_{Cell}=\set{l_1 \ldots l_m}$.
    \end{itemize}
    We take $x=l_1$.
\end{proposition}

\begin{proposition}\label{prop:lasso}
    The meta proposition~(\Cref{prop:meta}) for:
    Let $G$ be a lasso, with vertices $l_1,l_2, \ldots, l_m,c_1\ldots, c_{3k+r}$, where $0<m,k\in \Nat$ and $r<3$. Vertices 
    $l_1,l_2,\ldots, l_m, c_e$ form a path for some $1 \leq e \leq 3k+r$, while $c_1,c_2,\ldots c_{3k+r}$ form a cycle.  
    \begin{itemize}
        \item $C_{Pos}=\set{c_i|i=3j+1, j<k} \cup \set{c_i|i>3k}$,
        \item $C_{Neg}=\set{c_i|i=3j+3, j<k}$,
        \item $C_{Cell}=\set{c_i|i=3j+2, j<k}\cup \set{l_i|i\leq m}$.
    \end{itemize}
    We take $x=l_1$. It is not important to which node of the cycle the vertex $l_m$ is connected, i.e., what the value of $e$ is.
\end{proposition}

The proofs of~\Cref{prop:lasso,prop:Y} are similar to~\Cref{prop:cycle}. Although they are tricky, they were moved to~\Cref{app:sec_six}.

\begin{proof}[Proof of~\Cref{clm:last_claim}]

Our intention is to use some bad counters of the transfer VASS as memory cells, and use Propositions~\ref{prop:cycle},~\ref{prop:Y},~\ref{prop:lasso} to encode test transitions of the automaton with memory. Thus, we need to define subsets
$C_{Pos}, C_{Neg}, C_{Cell}$ of vertices in $G$ in such a way that they are consistent with Propositions~\ref{prop:cycle},~\ref{prop:Y},~\ref{prop:lasso}. It is possible as the graph $G$ is a clean-simple graph, so its good nodes can be coloured red and green, and we put green vertices to the set $C_{Pos}$ while red to $C_{Neg}$. Splitting bad vertices is trivial.

The number of vertices in $C_{Cell}$ is not smaller than $\frac{n}{5}$. This is because in $C_{Cell}$ we have all bad vertices except some bad vertices of type $3$, for which at most $2$ on every cycle may be outside $C_{Pos}\cup C_{Neg} \cup C_{Cell}$ and $\frac{1}{3}$ of remaining are in $C_{Cell}$. $\frac{1}{5}$ is the worst case when all bad nodes are on cycles of length $5$. 

We identify each cell with a vertex in $C_{Cell}$.
Recall that $M=(n,Q\times{\set{?,-1,1}}\times \set{1\ldots n},T,c_0,c_f)$. The set of states of the transfer VASS is
\[
Q\times{\set{?,-1,1}}\times \set{1\ldots n}\cup 
\bigcup_{t\in T\text{-}\mathrm{test}} Q_t
\] 
where $T\text{-}\mathrm{test}$ is the set of testing transitions and
$Q_t$ are states of transfer VASSes defined by Propositions~\ref{prop:cycle},~\ref{prop:Y}, and~\ref{prop:lasso}. 

Transitions are then translated directly. Instead of a testing transition for a memory cell $c$ we look to a corresponding counter and put a testing gadget provided by one of \Cref{prop:cycle,prop:Y}, and \Cref{prop:lasso}.
For an update transition for the cell $c$, we update the corresponding counter in the set $C_{Cell}$.

The initial and final markings are produced from $c_0$ and $c_f$ 
by keeping the state, setting all counters outside $C_{Cell}$ to $0$, and assigning counters in $C_{Cell}$ according to the corresponding cell values. \qedhere
\end{proof}

\section{Conclusion}

We have identified structural properties of transfer graphs
that determine the complexity of integer reachability in
single-transfer VASS. We showed that the problem belongs to NP
when the number of bad vertices is bounded, by establishing
a corresponding bound on the number of summary functions.
Conversely, we proved PSPACE-hardness for every polynomial-time
constructible family of transfer graphs containing sufficiently
many bad vertices. Together, these results characterize
the transfer structures responsible for the complexity gap
between NP and PSPACE.




\bibliography{refs}

\appendix

\section{Missing proofs from~\Cref{sec:four}}\label{app:sec_four}

\begin{proof}[Proof of~\Cref{lem:simpleTransferEq}]
We simply encode each transition with a complicated transfer 
function as a sequence of transitions with simple transfer functions. 
The only tricky part is if the complicated transfer contains a cyclic transfer, i.e., $c_1\to c_2, c_2\to c_3, c_3\to c_1$. Here, 
there is no order in which we can do simple transfers, so that the effect is the same. That is why we introduce the additional counter 
$c_{+}$, and do $c_1\to c_{+}$, next simple transfers $c_{3}\to c_{1}$, $c_2\to c_3$, and finally $c_{+}\to c_2$.
If the original cycle was of length $2$, its vertices could still be good in the transfer graph. 
But after the transformation the cycle involves one additional 
counter, its length grows to $3$ and all the involved vertices become bad of type $1$. \qedhere   
\end{proof}

\section{Missing proofs from~\Cref{sec:six}}\label{app:sec_six}

\begin{proof}[Proof of~\Cref{prop:Y}]
This construction relies on a subroutine $ShiftY(l_1,\ldots,l_m, j)$.
The subroutine has the following properties:
\begin{enumerate}
\item If the configuration is invalid, then after executing the subroutine it remains invalid.
\item The subroutine makes one nondeterministic guess of the value stored in counter $l_m$; if the guess is incorrect, then the configuration reached at the end of the subroutine is invalid.
\item If the nondeterministic guess is correct, then
the vector of counters reached at the end is the same as the initial one, except that the path $l_1,l_2,\ldots,l_m$ is cyclically shifted by one, i.e., counter $l_1$ gets the value of $l_m$, $l_2$ gets the value of $l_1$, and so on.
\end{enumerate}

The idea for the implementation of $ShiftY(l_1,\ldots,l_m, j)$ is as follows:
first, we guess whether the value stored in $l_m$ is $1$ or $-1$,
and depending on this guess we choose one of two branches of the control automaton. 

If we guess that it is $1$, then we subtract $1$ from counter $l_m$ and transfer this value to $c_l$.
Observe that if our guess is correct, the value of $c_l$ is unchanged, whereas if it is incorrect,
the value of $c_l$ becomes negative, so the vector of counters ceases to be valid.
The value stored in $l_m$ is now $0$, so we can perform a sequence of transfers 
$l_{m-1} \to l_{m}, l_{m-2}\to l_{m-1}, \ldots, l_{2}\to l_{3}, l_1\to l_2$.
Finally, since we guessed the value that was stored in $l_m$, we can restore it in $l_1$, i.e., we add $1$ to $l_1$.
This way we obtain the desired shift if the guess was correct, or we reach an invalid configuration if the guess was incorrect.    

If we guess that it is $-1$, then the construction is analogous, but we subtract $-1$ from $l_m$ and transfer the result to $c_r$.

To formalize the description given above, the control automaton (for $j \neq m$, for $j = m$ the final states differ) consists of two paths: $p_{j, init}, p_{j, 1}^{1}, \ldots, p_{j, m+1}^{1}, p_{j+1, init}$
or $p_{j, init}, p_{j, 1}^{-1}, \ldots, p_{j, m+1}^{-1}, p_{j+1, init}$ and the whole procedure can be written as below.

\begin{algorithmic}[1]
\Procedure{ShiftY}{$l_1,\ldots, l_m, j$}
    \State nondeterministically choose $a\in\set{-1,1}$
    \State $(p_{j, init}, \id, l_m \gets l_m - a, p_{j,1}^{a})$ 
    \If{$a=-1$}
        \State $(p_{j, 1}^{-1}, l_m\ \textbf{transfer to}\ c_r, \mathbf{0}, p_{j, 2}^{-1})$
    \Else
        \State $(p_{j, 1}^{1}, l_m\ \textbf{transfer to}\ c_l, \mathbf{0}, p_{j, 2}^{1})$
    \EndIf
    \For {$i \gets m\ \textbf{to}\ 2$}
        \State $(p_{j, m-i+2}^{a}, l_{i-1}\ \textbf{transfer to}\ l_i, \mathbf{0}, p_{j, m-i+3}^{a})$
    \EndFor
    \If{$j = m$}
        \State $(p_{j, m+1}^{a}, \id, l_1 \gets l_1 + a, p_{a})$
    \Else
        \State $(p_{j, m+1}^{a}, \id, l_1 \gets l_1 + a, p_{j+1, init})$
    \EndIf
\EndProcedure
\end{algorithmic}

Now, our intention is to apply the subprocedure $ShiftY(l_1, \ldots, l_m, j)$ for $1 \leq j \leq m$.
Observe that (assuming all the guesses are correct) before the last iteration, the original contents of $l_1$ have been transferred to $l_m$,
so the guess in the last iteration reflects the original content of $l_1$.
Thus at the end of the last iteration, based on $a$, we go to either $p_{-1}$ or $p_{1}$.
We also begin by setting $p_{init} = p_{1, init}$.

The whole procedure is now just a for-loop:
\begin{algorithmic}
    \For {$j \gets 1\ \textbf{to}\ m$}
        \State $ShiftY(l_1,\ldots,l_m, j)$
    \EndFor
\end{algorithmic}

Notice that in the control space we have to store only the last guessed value, thus omitting an exponential blow-up. \qedhere
\end{proof}

\begin{proof}[Proof of~\Cref{prop:lasso}]
    This proof is the most complicated one as it combines the main ideas of both~\Cref{prop:cycle} and~\Cref{prop:Y}.
    The procedure requires a subprocedure
    \hfill \break $ShiftLasso((l_1,l_2,\ldots l_m),(c_1,c_2\ldots c_n), e, j)$ where $e\leq n$ is the index of $c_e$, the vertex on the cycle to which $l_m$ is connected.
    The subprocedure is a direct counterpart to $ShiftY$ from \Cref{prop:Y} and has the following properties:
    \begin{enumerate}
        \item If the configuration is invalid, then after executing it, it remains invalid.
        \item The subprocedure nondeterministically guesses the value stored in counter $l_m$, and if the guess is incorrect then the configuration reached at the end of the subprocedure is invalid.
        \item If the nondeterministic guess is correct, then the subprocedure has one nondeterministic guess of the value stored in counter $c_2$, and if the guess is incorrect then the configuration reached at the end of the subprocedure is invalid.
        \item If the nondeterministic guess is correct, then
        the vector of counters reached at the end is the same as the initial one; only the path $l_1,l_2,\ldots l_m$ is cyclically shifted i.e., counter $l_1$ gets the value of $l_m$, $l_2$ gets the value of $l_1$ and so on.
    \end{enumerate}

    Again, the main algorithm will be repeating $ShiftLasso(\ldots)$ for $1 \leq j \leq m$,
    reaching the same content of the counters.
    At the end of the last iteration we will go to $p_j$ where $j$ is the content of the counter $l_1$.

    Conceptually $ShiftLasso(\ldots)$ is performed in $5$ steps.
    First, we guess whether the value $a$ of $l_m$ is $1$ or $-1$. Depending on this, the control automaton enters one of two branches.
    Also, if we guess that the value of $l_m$ is $1$ then we subtract $1$ from this counter, and if we guess that the value of $l_m$ is $-1$ then we add $1$ to this counter.
    So, after this operation the value of $l_m$ is $-2,0$, or $2$; it is $0$ if the guess was correct and nonzero if it was incorrect. 
    
    The second step is to cyclic shift the cycle in such way that in $c_e$ we have the value of $c_3$ or $c_1$, depending on the choice in the first step.
    If it was $1$, then we shift to $c_3$ and if it was $-1$ then we shift to $c_1$.
    Remember that $c_1\in C_{Pos}$ and $c_3\in C_{Neg}$.
    
    In the third step we transfer $l_m$ to $c_e$, which now stores the value of $c_3$ or $c_1$.
    Note that if the guess was correct the value of $c_e$ is unchanged.
    If it was incorrect then the value of $c_e$ is the original value of $c_3$ minus $2$ or the original value of $c_1$ plus $2$.
    
    In the fourth step we perform a chain of transfers,
    $l_{m-1}\to l_m$, $l_{m-2}\to l_{m-1},\ldots ,l_2\to l_3, l_1\to l_2$, and restore the guessed value in $l_1$.
    This implements the cyclic shift of the sequence $l_1,l_2,\ldots, l_m$. 
    
    What remains is to make a cyclic shift of the cycle so that 
    the value in $c_e$ goes back to the counter $c_1$ or $c_3$, 
    restoring the original distribution of values on the cycle $c_1,c_2, \ldots ,c_n$.
    This is correct under the assumption that the guess was correct; otherwise, we finish with an invalid vector of counter values.

As in the proof of~\Cref{prop:cycle}, when shifting the cycle, we guess the value of the cell counter
$c_2$ and temporarily make it zero.  We use $s$ to denote the current
position of this zero counter; this cursor is part of the control state and
is advanced after every elementary rotation. 

To rotate the cycle by one we use a subprocedure $RotateCycle(s, b, j, k,a)$, where
$j, k,a$ are used only to label states of the control automaton:
\begin{itemize}
    \item $j$ signifies the current iteration number of $ShiftLasso(\ldots)$.
    \item Inside $ShiftLasso(\ldots)$ we perform the rotations of the cycle many times, and $k$ remembers which rotation we are currently in.
\item $a$ is used to store the guessed value of $l_m$, and based on it we identify the branches of the control automaton.
\end{itemize}

As stated before, inside $RotateCycle(s,b,j,k,a)$ we temporarily make the counter $c_s$ zero and shift the cycle by one.
The transfer sequence depends on the guess of $c_s$ (which is called $b$): for the guess $-1$ it begins
with a transfer from $c_{s - 1}$ to $c_s$, whereas for the guess $1$ it begins with a transfer from $c_s$ to $c_{s+1}$.
Notice that to shift the cycle by one we only need $n-1$ transfers.

In $ShiftLasso(\ldots)$ the same guess $b$ will be retained in the control state during
all consecutive partial rotations and is restored only after the whole sequence of cycle rotations has
been performed. 

The following pseudocode makes the construction of $RotateCycle(\ldots)$ and $ShiftLasso(\ldots)$ explicit.  Indices of the
cycle are taken modulo $n$ (with representatives in $\set{1,\ldots,n}$).

\begin{algorithmic}[1]
\Procedure{RotateCycle}{$s,b,j,k,a$}
    \State $(p_{j,k,1}^{a,b,s}, \id, c_s \gets c_s - b, p_{j,k,2}^{a,b,s})$
    \If{$b=-1$}
        \For{$i \gets 1$ \textbf{to} $n-1$}
            \State $(p_{j,k,i+1}^{a,b,s}, c_{s-i}\ \textbf{transfer to}\ c_{s-i+1}, \mathbf{0}, p_{j,k,i+2}^{a,b,s})$
        \EndFor
        \State $(p_{j,k,n+1}^{a,b,s}, \id, c_{s+1} \gets c_{s+1} - 1, p_{j,k+1,1}^{a,b,s})$
    \Else
        \For{$i \gets 0$ \textbf{to} $n-1$}
            \State $(p_{j,k,i+2}^{a,b,s}, c_{s-i}\ \textbf{transfer to}\ c_{s-i+1}, \mathbf{0}, p_{j,k,i+3}^{a,b,s})$
            \EndFor
        \State $(p_{j,k,n+2}^{a,b,s}, \id, c_{s+1} \gets c_{s+1} + 1, p_{j,k+1,1}^{a,b,s})$
    \EndIf
\EndProcedure
\Statex

\Procedure{ShiftLasso}{$(l_1,\ldots,l_m),(c_1,\ldots,c_n),e, j$}
    \State \Comment{Step $1$.}
    \State nondeterministically choose $a\in\set{-1,1}$
    \State $(p_{j, init}, \id, l_m \gets l_m-a, p_{j, 1}^{a})$
        \Comment{$l_m=0$ exactly when the guess is correct}
    \If{$a=1$}
        \State $q\gets 3$
    \Else
        \State $q\gets 1$
    \EndIf
    \State \Comment{Step $2$.}
    \State $s \gets 2$ \Comment{We start rotating cycle from vertex 2}
    \State nondeterministically choose $b\in\set{-1,1}$
    \State $(p_{j, 1}^{a}, \id, \mathbf{0}, p_{j,1,1}^{a, b, s})$
    \For {$i \gets q\ \textbf{to}\ e-1$}
        \Comment {$\bmod n$}
        \State \Call{RotateCycle}{$s,b,j,(i-q \bmod n)+1,a$}
        \Comment{move the original content of $c_q$ to $c_e$}
        \State $s \gets s +1 \bmod n$ \Comment{ We update s as the cycle is rotated by $1$}

    \EndFor
    \State \Comment{the original value of $c_q$ has been rotated to $c_e$.}
    \State \Comment{Step $3$.}
    \State $(p_{j, (e-q\bmod n)+1, 1}^{a, b, s}, l_m\ \textbf{transfer to}\ c_e, \mathbf{0}, p_{j, shift, 1}^{a, b, s})$
    \State \Comment{Step $4$.}
    \For {$i \gets m \textbf{ to } 2$}
        \State $(p_{j, shift, m-i+1}^{a, b, s}, l_{i-1}\ \textbf{transfer to}\ l_{i}, \mathbf{0}, p_{j, shift, m-i+2}^{a, b, s})$
    \EndFor
    \State $(p_{j, shift, m}^{a, b, s}, \id, l_1 \gets l_1+a, p_{j, n+1,1}^{a, b, s})$
    \Comment{restore value of $l_m$ at the counter $l_1$}

    \State \Comment{ step 5}
    \For{$i \gets e\ \textbf{to}\ q-1 $}
        \Comment {$\bmod n$}
        \State \Call{RotateCycle}{$s, b, j, n+(e-i \bmod n)+1,a$}
        \State $s \gets s +1$
    \EndFor

    \If {$j \neq m$}
        \State $(p_{j,n+(q-e \bmod n)+1, 1}^{a, b, s}, \id, \mathbf{0}, p_{j+1, init})$
    \Else
        \State $(p_{j,n+(q-e \bmod n)+1, 1}^{a, b, s}, \id, \mathbf{0}, p_{a})$
    \EndIf
\EndProcedure
\end{algorithmic}

Note that for-loops rotating the cycle are taken modulo $n$. This means that if the starting value $y$ is greater
than the last value $z$ we actually perform a for loop from $y$ to $n+z$.

For $RotateCycle(s,b,j,k,a)$ the control automaton is just a path $p_{j, k, 1}^{a, b, s},\ldots, p_{j, k, n+1}^{a, b, s}, p_{j, k+1, 1}^{a,b,s}$
or $p_{j, k, 1}^{a, b, s},\ldots, p_{j, k, n+2}^{a, b, s}, p_{j, k+1, 1}^{a,b,s}$ depending on whether $b=1$ or $b=-1$, so that we are ready for the next execution of $RotateCycle$.

For $ShiftLasso(\ldots)$, the control automaton consists of four paths, defined by the choice of $a$ and $b$.
A single path goes through $p_{j, init}, p_{j, 1}^{a}, (p_{j,1,1}^{a,b,s}, \ldots, p_{j, (e-q\bmod n)+1, 1}^{a, b, s})$ --- first sequence of cycle rotations, then
$p_{j,shift,1}^{a,b,s}, \ldots, p_{j, shift, m}^{a, b, s}$ --- shifts on the path, then
$(p_{j, n+1,1}^{a, b, s}, \ldots, p_{j, n+(q-e\bmod n)+1,1}^{a,b,s})$ --- second sequence of cycle rotations and at the end
$p_{j+1, init}$ if $j < m$ or $p_{a}$ if $j = m$.
As in~\Cref{prop:Y} we begin with $p_{init} = p_{1, init}$. \qedhere
\end{proof}

\paragraph*{Declaration on the use of generative AI.}
The authors used a generative AI tool to assist in drafting the exposition of the
standard reduction in the proof of \Cref{thm:hardness_of_memory_automata}. The proof, its notation,
and all mathematical claims were subsequently checked, corrected, and verified by
the authors, who take full responsibility for the content.
\newpage

\section{Simple graphs}\label{app:two}

\begin{proof}[Proof of~\Cref{lem:simple_graph}]
We construct $G'$ by induction on the number of edges. The following claim will be useful in the construction.

\begin{claim}
If we can split $G$ into two parts $G_1$ and $G_2$ such that edges go only from $G_1$ to $G_2$, and all good vertices in $G_1$ are also good in $G$, then removing all edges going from $G_1$ to $G_2$ does not change the number of bad vertices.
\end{claim}
\begin{proof}
Note that if $x$ is bad and if we remove some edges from $G$, but without touching edges between nodes in $\neighbourhoodStar{x}$, then $x$ stays bad.
This is the case for every vertex in $G_2$.
Since all bad vertices in $G$ restricted to $G_1$ are bad in $G_1$, the number of bad vertices has not decreased. \qedhere
\end{proof}

So we are ready to start the induction. If $G$ is empty then the claim trivially holds. Suppose that the induction hypothesis holds for all graphs with less than $m$ edges. We will show that it holds for $G$ with $m$ edges.

We consider several cases:
\begin{enumerate}
    \item If $G$ is disconnected then we can apply the induction hypothesis to each connected component and take the union of the resulting simple subgraphs.
    \item Else, if $G$ has an outgoing edge from a good vertex, then it can be removed without changing the number of bad vertices.
    We can apply the induction hypothesis to the resulting graph.
    \item Else, if we can find a simple cycle with at least $3$ vertices in $G$, we take all vertices from which this cycle is reachable.
    The induced subgraph of those vertices is $G_1$. Note that all vertices in $G_1$ are bad of type $1$ or $3$, and there are no edges from $G_2=G - G_1$ to $G_1$.
    Thus, we can use our claim and remove all edges from $G_1$ to $G_2$. Now for $G_1$ we can keep only the simple cycle and edges witnessing that from all other vertices it is possible to reach the cycle.
    This subgraph is simple and all nodes within it are bad.
    For $G_2$ we can apply the induction hypothesis and get a simple subgraph with at least $\frac{1}{3}|B(G_2)|$ bad vertices.
    The union of the simple subgraph from $G_1$ and the simple subgraph from $G_2$ is a simple subgraph of $G$ with at least $\frac{1}{3}|B(G)|$ bad vertices.
    \item Else, if in $G$ there is a bad vertex of type $4$ with two outgoing edges then one of them can be removed without changing the number of bad vertices.
    This is because there are no bad vertices of type $1$ or $3$.
    If we remove the correct edge, then this vertex will stay bad, and all other vertices of type $4$ from which it was reachable stay bad as well.
    We can apply the induction hypothesis to the resulting graph. 
    \item Else if in $G$ there is a bad vertex $x$ of type $2$ with at least two good out-neighbours $z$ and $y$,
    then as $G_1$ we take a subgraph of $G$ induced by all vertices from which $x$ can be reached plus $y$ and $z$.
    $G_2 = G - G_1$. One more time we can use the claim and separate $G_1$ from $G_2$.
    For $G_2$ we can apply the induction hypothesis and get a simple subgraph with at least $\frac{1}{3}|B(G_2)|$ bad vertices.
    For $G_1$ we simply remove edges from $x$ other than to $y$ and $z$.

    All the vertices other than $x, y, z$ are bad of type $4$ and by the previous item have one out-neighbour each,
    $x$ has two out-neighbours and $y, z$ have none.
    Thus, this trimmed $G_1$ is simple and all vertices except for $y$ and $z$ are bad in it.
    The union of the simple subgraph from $G_1$ and the simple subgraph from $G_2$ is a simple subgraph of $G$ with at least $\frac{1}{3}|B(G)|$ bad vertices.
    \item Else:
    \begin{claim}\label{claim:sss_of_bad_2}
    Note that if from $x$, a bad vertex of type $2$, it is possible to reach $y$, another bad vertex of type $2$, then $x$ and $y$ are in the same strongly connected component.
    \end{claim} 
    \begin{proof}
    Indeed, if on the path from $x$ to $y$ there is a vertex of type $4$ then $x$ is also of type $4$ (no cycles of length at least $3$).
    But by item $4$ bad vertices of type $4$ have one out-neighbour, contradiction with $x$ being bad of type $2$.
    Thus all vertices on the path from $x$ to $y$ are of type $2$.

    If $x$ and $y$ are in different strongly connected components then there are two vertices $x',y'$ on this path such that $y'\in \neighbourhoodOpen{x'}$ but $x'\notin \neighbourhoodOpen{y'}$.
    But in this case $x'$ is of type $4$ (we established that all vertices on the path, so $y'$ as well, are bad of type $2$) so we have a contradiction. \qedhere
    \end{proof}

    As a corollary we have:
    \begin{claim}\label{claim:two_way_edges}
    If $x$ and $y$ are bad vertices of type $2$ and there is an edge $(x, y)$ then there is an edge $(y, x)$.
    \end{claim}
    \begin{proof}
    By~\Cref{claim:sss_of_bad_2} we know that $x$ and $y$ belong to the same strongly connected component,
    so there is a simple path from $y$ to $x$.
    That path cannot have length at least $2$, otherwise we would have a cycle of length at least $3$.
    Thus, it has to be a single edge from $y$ to $x$. \qedhere
    \end{proof}

    We consider further subcases:
    \begin{enumerate}
        \item In $G$ there is a vertex with out-degree at least $3$ and out-neighbours $u_1,u_2,\ldots , u_l$, where $l\geq 3$.
        We know that it has to be of type $2$.
        An example can be seen in~\Cref{fig:out_deg3}.
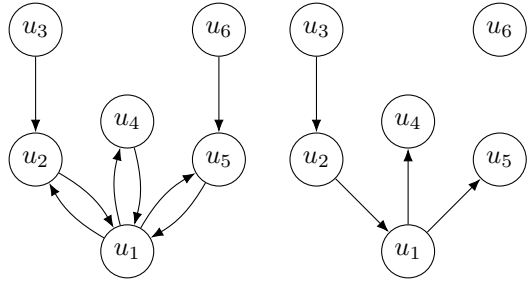
\begin{figure}[t]
\centering
\begin{tikzpicture}[>=Latex, node distance=10mm, every node/.style={draw,circle,minimum size=7mm,inner sep=0pt}]
\node (u1) {$u_1$};
\node[above left=of u1] (u2) {$u_2$};
\node[above=of u2] (u3) {$u_3$};
\node[above=of u1] (u4) {$u_4$};
\node[above right=of u1] (u5) {$u_5$};
\node[above=of u5] (u6) {$u_6$};

\draw[->] (u3) to (u2);
\draw[->] (u6) to (u5);

\draw[->, bend left=15] (u1) to (u2);
\draw[->, bend left=15] (u2) to (u1);

\draw[->, bend left=15] (u1) to (u4);
\draw[->, bend left=15] (u4) to (u1);

\draw[->, bend left=15] (u1) to (u5);
\draw[->, bend left=15] (u5) to (u1);

\node[right=3cm of u1] (u1') {$u_1$};
\node[above left=of u1'] (u2') {$u_2$};
\node[above=of u2'] (u3') {$u_3$};
\node[above=of u1'] (u4') {$u_4$};
\node[above right=of u1'] (u5') {$u_5$};
\node[above=of u5'] (u6') {$u_6$};

\draw[->] (u3') to (u2');

\draw[->] (u2') to (u1');

\draw[->] (u1') to (u4');

\draw[->] (u1') to (u5');

\end{tikzpicture}
\caption{Example graph (left) with a bad vertex of type $2$ with out-degree greater than $2$ and the graph's simplification (right)}.
\label{fig:out_deg3}
\end{figure}

        We denote by $Pre(u_i, u_j,x)$ the set of vertices from which $x$ can be reached without going through $u_i$ and $u_j$.
        Let $G_1$ be the subgraph of $G$ induced by all vertices from which $x$ can be reached and $G_2=G-G_1$.
        Then there are $i\neq j$ in $\set{1,2,\ldots , l}$ such that $Pre(u_i,u_j,x)$ contains at least $\frac{1}{3}$ of vertices in $G_1$.
        We handle $G_2$ using the induction hypothesis.

        For a simplification of $G_1$ we retain edges $(x,u_i), (x,u_j)$ and edges needed to witness that from every node in $Pre(u_i,u_j,x)$ it is possible to reach $x$, without going through $u_i$ and $u_j$.
        The resulting graph is simple and has at least $\frac{1}{3}|B(G_1)|$ bad vertices.

        \item Else $G$ consists of a strongly connected component of vertices of type $2$, some attached trees of vertices of type $4$ and good vertices.
        \begin{claim}
        Note that a vertex of type $4$ cannot be a part of the strongly connected component containing vertices of type $2$.
        \end{claim}
        \begin{proof}
        Indeed, if that was the case then in the strongly connected component we would have a vertex $x$ of type $4$ with an edge to a vertex $y$ of type $2$ in the same strongly connected component.
        Since there are no simple cycles of length at least $3$, then there would be an edge $(y, x)$.

        By the previous item, we can assume that every vertex of type $2$ has exactly two out-neighbours.

        By definition of type $4$, we know that from $x$ we can reach some bad vertex $z$ (potentially $y$) of type $2$, such that two out-neighbours of $z$
        do not belong to the path from $x$ to $z$.
        
        But this is impossible. Because of~\Cref{claim:sss_of_bad_2} we can assume that all vertices on the path from $x$ to $z$ are of type $2$.
        The last edge of the path goes from some $z'$ to $z$ and because of~\Cref{claim:two_way_edges} there is an edge $(z,z')$ in the graph,
        thus $z$ does not have two out-neighbours not on the path from $x$ to $z$. \qedhere
        \end{proof}

        From \Cref{claim:sss_of_bad_2} and since every vertex of type $4$ is connected to only one strongly connected component of vertices of type $2$,
        the only way that the strongly connected components can be connected (in the undirected sense) is through edges to common good vertices.
        The whole graph looks as in \Cref{fig:type2-scc}.

\begin{figure}[t]
\centering
\begin{tikzpicture}[>=Latex, node distance=10mm, every node/.style={draw,circle,minimum size=7mm,inner sep=0pt}]
\node (u1) {$u_1$};
\node[right=of u1] (u2) {$u_2$};
\node[right=of u2] (u3) {$u_3$};
\node[right=of u3] (u4) {$u_4$};
\node[above=of u1] (t11) {$t_{1,1}$};
\node[above=of u2] (t21) {$t_{2,1}$};
\node[above=of u3] (t31) {$t_{3,1}$};
\node[above=of u4] (t41) {$t_{4,1}$};

\node[below=of u1] (g1) {$g_{1}$};
\node[below right=of u4] (g4) {$g_{4}$};

\node[above=of t11] (t12) {$t_{1,2}$};
\node[above left=of t31] (t32) {$t_{3,2}$};
\node[above right=of t31] (t33) {$t_{3,3}$};
\node[above=of t33] (t34) {$t_{3,4}$};

\node[right=of u4] (u5) {$u_5$};
\node[right=of u5] (u6) {$u_6$};

\node[above=of u5] (t51) {$t_{5,1}$};
\node[above=of u6] (t61) {$t_{6,1}$};

\draw[->, bend left=30] (u1) to (u2);
\draw[->, bend left=30] (u2) to (u1);

\draw[->, bend left=30] (u2) to (u3);
\draw[->, bend left=30] (u3) to (u2);

\draw[->, bend left=30] (u3) to (u4);
\draw[->, bend left=30] (u4) to (u3);

\draw[->] (u1) -- (g1);
\draw[->] (u4) -- (g4);

\draw[<-] (u1) -- (t11);
\draw[<-] (u2) -- (t21);
\draw[<-] (u3) -- (t31);
\draw[<-] (u4) -- (t41);

\draw[<-] (t11) -- (t12);
\draw[<-] (t31) -- (t32);
\draw[<-] (t31) -- (t33);

\draw[<-] (t33) -- (t34);

\draw[->, bend left=30] (u5) to (u6);
\draw[->, bend left=30] (u6) to (u5);

\draw[<-] (u5) -- (t51);
\draw[<-] (u6) -- (t61);

\draw[->] (u5) -- (g4);
\draw[->] (u6) -- (g4);

\end{tikzpicture}
\caption{Two strongly connected components of bad vertices of type $2$. Vertices labelled with $u_?$ are of type $2$. We also have attached trees of bad vertices of type $4$ labelled with $t_{?,?}$. Vertices $g_1,g_4$ are good.}
\label{fig:type2-scc}
\end{figure}
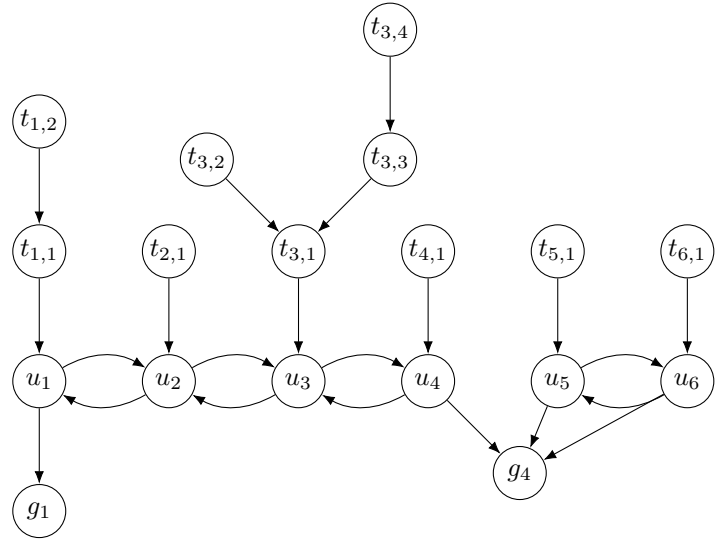

        We order bad vertices of type $2$ from left to right $u_1,u_2\ldots u_k$.
        Every vertex $u_i$ of type $2$ we label with $l_i$ --- the number of vertices of type $4$ from which $x$ can be reached without going through any other $u$-s plus $1$ for $u_i$. 

        Observe that one of the sums
        $\sum_{i\leq k, i\text{ is odd}} l_{i}$ or $\sum_{i\leq k, i\text{ is even}}  l_{i}$ is at most $\frac{1}{2}|\bad(G)|$. 

        We pick the smaller sum and remove all edges outgoing from the vertices of type $2$ that are in this sum. Next we remove outgoing edges from all vertices that start to be good because of this operation.
        The resulting graph is simple and has at least $\frac{1}{2}|\bad(G)|$ bad vertices.
        
        The result of this operation performed on \Cref{fig:type2-scc} can be seen in \Cref{fig:simplified-type2-scc}.
    
\begin{figure}[t]
\centering
\begin{tikzpicture}[>=Latex, node distance=10mm, every node/.style={draw,circle,minimum size=7mm,inner sep=0pt}]
\node (u1) {$u_1$};
\node[right=of u1] (u2) {$u_2$};
\node[right=of u2] (u3) {$u_3$};
\node[right=of u3] (u4) {$u_4$};
\node[above=of u1] (t11) {$t_{1,1}$};
\node[above=of u2] (t21) {$t_{2,1}$};
\node[above=of u3] (t31) {$t_{3,1}$};
\node[above=of u4] (t41) {$t_{4,1}$};

\node[below=of u1] (g1) {$g_{1}$};
\node[below right=of u4] (g4) {$g_{4}$};

\node[above=of t11] (t12) {$t_{1,2}$};
\node[above left=of t31] (t32) {$t_{3,2}$};
\node[above right=of t31] (t33) {$t_{3,3}$};
\node[above=of t33] (t34) {$t_{3,4}$};

\node[right=of u4] (u5) {$u_5$};
\node[right=of u5] (u6) {$u_6$};

\node[above=of u5] (t51) {$t_{5,1}$};
\node[above=of u6] (t61) {$t_{6,1}$};

\draw[->] (u1) to (u2);

\draw[->] (u3) to (u2);

\draw[->] (u3) to (u4);

\draw[->] (u1) -- (g1);

\draw[<-] (u1) -- (t11);
\draw[<-] (u3) -- (t31);

\draw[<-] (t11) -- (t12);
\draw[<-] (t31) -- (t32);
\draw[<-] (t31) -- (t33);

\draw[<-] (t33) -- (t34);

\draw[<-] (u6) -- (u5);

\draw[<-] (u5) -- (t51);

\draw[->] (u5) -- (g4);

\end{tikzpicture}
\caption{Result of the simplification of~\Cref{fig:type2-scc}, retaining at least half of the bad vertices}.
\label{fig:simplified-type2-scc}
\end{figure}
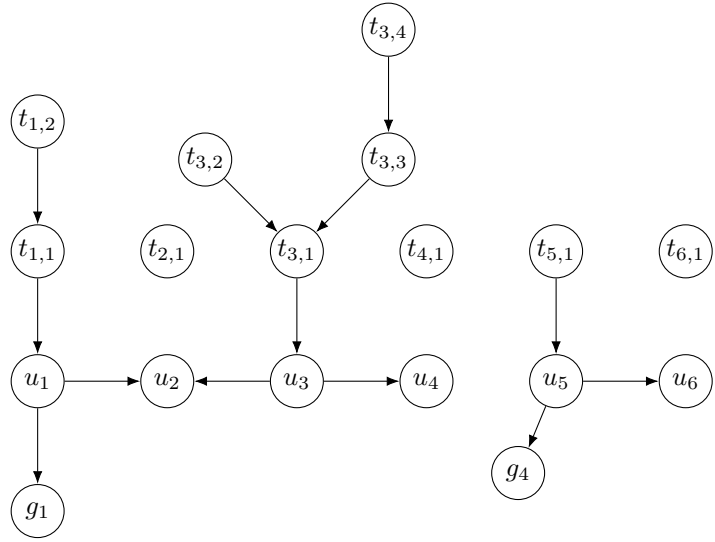
     \qedhere
    \end{enumerate}
\end{enumerate}
\end{proof}

\begin{proof}[Proof of~\Cref{lem:clean_simple_graph}]
We proceed by induction on the number of bad vertices. The lemma trivially holds when there are no bad vertices.
Assume the lemma holds for all simple graphs with at most $n-1$ bad vertices; we prove it holds for simple graphs with $n$ bad vertices.

Suppose $G$ contains two connected components with at least one bad vertex each.
We can apply the induction hypothesis to each component and merge the solutions.
Thus, without loss of generality, we may assume $G$ has a single connected component containing all bad vertices (possibly with some isolated good vertices, which do not affect the construction).
If this component contains only bad nodes of type $1$ or $3$, then the graph is already clean and simple.
Therefore, we focus on the case where the graph contains only bad vertices of types $2$ and $4$, along with some good vertices.
Note that by a definition of a simple graph, a single connected component cannot have bad vertices of types $1$ or $3$ and $2$ or $4$ simultaneously.

The graph $G'$ is produced by removing edges outgoing from certain bad vertices of type $2$,
thereby converting them to good vertices.
As this can also change some bad vertices of type $4$ into good ones, we also remove edges outgoing from them,
since, in simple graphs, good vertices have no outgoing edges.
Note that removing edges from a single vertex of type $2$ can substantially reduce the number of bad vertices in the entire graph,
as many previously bad vertices of type $4$ may be transformed to good ones by this operation.

To identify which vertices of type $2$ should have their outgoing edges removed, we use the following construction.
We build a weighted graph $H$ whose vertices are the good vertices of $G$.
Each edge in $H$ is labelled with a natural number representing a weight.
For every pair of vertices $x,y$ in $H$, we create an edge weighted as follows: if $z$ is a bad vertex of type $2$ in $G$ with exactly two good out-neighbours $x$ and $y$,
we add to the weight of the edge $\{x,y\}$ the value $1$ plus the number of type $4$ vertices from which $z$ is reachable without going through other vertices of type $2$.

We seek to colour the vertices of $H$ with red and green such that the sum of weights of edges with endpoints of different colours is at least $\frac{1}{2}$ of the total weights sum.
Once such a colouring is found, we remove outgoing edges from bad vertices of type $2$ whose out-neighbours are all monochromatic (same colour), as described $2$ paragraphs earlier.
This produces a clean-simple subgraph with at least $\frac{n}{2}$ bad vertices, as required.

To compute such a colouring of $H$, we use the following algorithm.
We initialize all vertices of $H$ to green.
We then repeatedly apply the following subroutine: for each vertex,
compute the sum of weights of incident edges with endpoints of the same colour and the sum of weights of incident edges with endpoints of different colours.
If the latter sum is smaller for a vertex $x$, recolour this vertex. We continue until no further improvement is possible.

Note that this operation increases the sum of weights of edges with bi-chromatic endpoints.
Indeed, all edges that were monochromatic adjacent to $x$ become bi-chromatic, while all bi-chromatic edges adjacent to $x$ become monochromatic.
Since the sum of monochromatic edge weights incident to $x$ is greater than the sum of bi-chromatic edge weights incident to $x$, this operation increases the total sum of bi-chromatic edge weights.

When no further improvement is possible, every vertex in $H$ has a sum of bi-chromatic incident edge weights greater than or equal to the sum of monochromatic incident edge weights.
Therefore, the same property holds for the entire graph $H$, completing the proof. \qedhere
\end{proof}

\begin{proof}[Proof of~\Cref{thm:hardness_of_memory_automata}]
    We give a polynomial-time reduction from the problem
\textsc{Space-Bounded-TM-Acceptance}, which is PSPACE-complete by
\Cref{thm:PSPACE_Hardness}.

Let
$
\mathcal{M} =
(Q,\set{0,1},\delta,q_0,q_{\mathrm{acc}},q_{\mathrm{rej}})
$ and $(1^n)$ be an instance of
\textsc{Space-Bounded-TM-Acceptance}.
Without loss of generality, we assume that the machine uses the tape
cells $(1,\ldots,n)$, that its head initially points to cell $1$,
and that any transition leaving this range leads to rejection.

\newcommand{\enc}{\operatorname{enc}}
We construct an automaton with memory
\[
\mathcal{A}=
(n,Q_{\mathcal{A}}\times(\set{?}\cup\set{-1,1})
\times\set{1,\ldots,n},T,c_0,c_f).
\]
The value $(-1)$ represents the tape symbol $(0)$, whereas $(1)$
represents the tape symbol $(1).$ Formally, let
$
\enc(0)=-1
\qquad\text{and}\qquad
\enc(1)=1.
$ The set of control states is
$
Q_{\mathcal{A}}=
Q\cup\set{r_1,\ldots,r_n,q_f},
$
where $(r_1,\ldots,r_n,q_f)$ are fresh states used to erase the
memory after acceptance.

The initial configuration is
$
c_0=
\bigl((q_0,?,1),(-1,\ldots,-1)\bigr),
$
since the empty input word is represented by an all-zero tape.
The final configuration is
$
c_f=
\bigl((q_f,?,1),(-1,\ldots,-1)\bigr).
$

{\bf Simulation of Turing-machine transitions.}
Consider a transition
$
(q,a,q',b,D)\in\delta.
$
For every tape position $(i\in\set{1,\ldots,n})$, let
\[
i_D=
\begin{cases}
i-1, & \text{if }D=L,\\
i+1, & \text{if }D=R.
\end{cases}
\]
Whenever $(i_D\in\set{1,\ldots,n})$, we add the following two
transitions to $(T)$:
$
(q,?,i)
\xrightarrow{\enc(a)}
(q,\enc(a),i),
$
and
$
(q,\enc(a),i)
\xrightarrow{\enc(b)}
(q',?,i_D).
$
The first transition tests whether the currently scanned tape cell
contains $a$. The second transition writes $b$, changes the
control state to $(q')$, and moves the head according to $D$.

Thus, every Turing-machine step is represented by exactly one test
transition followed by one update transition.

{\bf Cleanup after acceptance.}
The accepting state of the Turing machine may be reached with an
arbitrary tape content. Since reachability asks for one fixed final
configuration, we add a cleanup phase that resets all memory cells to
(-1).

For every $(i\in\set{1,\ldots,n})$ and $(j\in\set{-1,1})$, add
$
(q_{\mathrm{acc}},?,i)
\xrightarrow{j}
(q_{\mathrm{acc}},j,i)
$
and
$
(q_{\mathrm{acc}},j,i)
\xrightarrow{j}
(r_1,?,1).
$
These transitions preserve the currently scanned memory value and
move the memory head to cell (1).

Next, for every $(i\in\set{1,\ldots,n})$ and every
$(j\in\set{-1,1})$, add
$
(r_i,?,i)
\xrightarrow{j}
(r_i,j,i).
$
For (i<n), add
$
(r_i,j,i)
\xrightarrow{-1}
(r_{i+1},?,i+1),
$
and, for (i=n), add
$
(r_n,j,n)
\xrightarrow{-1}
(q_f,?,1).
$
Hence, after entering $(q_{\mathrm{acc}})$, the automaton visits all
memory cells in order, resets each of them to $(-1)$, and finally
reaches $(c_f)$.

{\bf Correctness.}
We prove that
$
\mathcal{M}
\text{ accepts the empty word using at most }n\text{ cells}
\quad\Longleftrightarrow\quad
c_0\to_{\mathcal{A}}^* c_f.
$

Suppose first that $(\mathcal{M})$ accepts. Initially, the memory
vector of $(\mathcal{A})$ represents the all-zero tape, and the
memory head is at position $(1).$ By construction, every transition
of $(\mathcal{M})$ is simulated by one enabled test transition and
one enabled update transition of $(\mathcal{A}).$ Therefore,
$(\mathcal{A})$ reaches a configuration of the form
$
\bigl((q_{\mathrm{acc}},?,i),v\bigr)
$
for some (i) and some memory valuation (v). The cleanup phase is
then enabled and leads to $(c_f).$ Hence, $
(c_0\to_{\mathcal{A}}^* c_f).$

Conversely, suppose that
$
c_0\to_{\mathcal{A}}^* c_f.
$
The only transitions entering $(q_f)$ belong to the cleanup phase.
Therefore, before reaching $(q_f)$, the automaton must reach a state
of the form $((q_{\mathrm{acc}},?,i)).$ Before entering
$(q_{\mathrm{acc}})$, the automaton can only execute transitions
simulating transitions of $(\mathcal{M}).$ Thus, the corresponding
sequence of simulated Turing-machine transitions reaches
$(q_{\mathrm{acc}}).$ Hence, $(\mathcal{M})$ accepts the empty word
using at most $(n)$ tape cells.

Finally, the construction contains
$(O(n\cdot |\delta|+n))$ transitions and is computable in polynomial
time, because (n) is given in unary. Therefore, this is a
polynomial-time many-one reduction from
\textsc{Space-Bounded-TM-Acceptance} to reachability for automata
with memory. Consequently, the latter problem is PSPACE-hard. \qedhere
\end{proof}

\end{document}